\documentclass{article}

\usepackage{flushend}
\usepackage{lineno,hyperref}
\usepackage{epsfig}
\usepackage{graphicx}%hyperref amssymb pifont lscape subfigure multirow
\usepackage{color}
\usepackage{subfigure}
\usepackage{graphics}% Include figure files
\usepackage{dcolumn}% Align table columns on decimal point
\usepackage{bm}% bold math
\usepackage{multirow}
\usepackage{longtable}
\usepackage[ruled, lined]{algorithm2e}        %[linesnumbered,ruled,vlined]
\usepackage{setspace}
\usepackage{amsthm}
\usepackage{mathabx}
\usepackage{textcomp}
\usepackage[table]{xcolor}
\usepackage{tabularx}
\usepackage{changepage}

\newtheorem{lemma}{Lemma}
\newtheorem{theorem}{Theorem}
\newtheorem{prop}{Proposition}

\SetAlFnt{\small}
\SetAlCapFnt{\small}
\SetAlCapNameFnt{\small}
\SetAlCapHSkip{0pt}
\IncMargin{-\parindent}

\begin{document}

%\begin{frontmatter}
\textbf{
\large {\centerline{
A system architecture for parallel analysis of flux-balanced }}}
\large{\textbf{\centerline{metabolic pathways}}}
%Elementary Flux Mode Analysis: A Graph-Based Approach
\vspace{1em}

\small{
\centerline{Mona Arabzadeh$^1$, Mehdi Sedighi$^1$, Morteza Saheb Zamani$^1$,}
\centerline{and Sayed-Amir Marashi$^2$}}
\vspace{1em}
\footnotesize{
\centerline{$^1$Department of Computer Engineering and Information Technology,}
\centerline{Amirkabir University of Technology, Tehran, Iran and}
\vspace{0.2em}
\centerline{$^2$Department of Biotechnology, College of Science, University of Tehran, Tehran, Iran.}
\vspace{1em}
\centerline{\{m.arabzadeh,msedighi,szamani\}@aut.ac.ir, marashi@ut.ac.ir}
}

\begin{abstract}
Elementary flux mode (EFM) analysis is a well-studied method in constraint-based modeling of metabolic networks. In EFM analysis, a network is decomposed into minimal functional pathways based on the assumption of balanced metabolic fluxes.
In this paper, a system architecture is proposed that approximately models the functionality of metabolic networks. The AND/OR graph model is used to represent the metabolic network and each processing element in the system emulates the functionality of a metabolite. The system is implemented on a graphics processing unit (GPU) as the hardware platform using CUDA environment.
The proposed architecture takes advantage of the inherent parallelism in the network structure in terms of both pathway and metabolite traversal. The function of each element is defined such that it can find flux-balanced pathways.
Pathways in both small and large metabolic networks are applied to the proposed architecture and the results are discussed.
\end{abstract}

%\begin{keyword}

%\end{frontmatter}

%\linenumbers
\normalsize
%%---------------------------------------------------------------------------------------------------------------------------------------------
\section{Introduction}\label{sec:introduction}
\emph{System architecture} is defined as a generic discipline to handle ``systems'' considered as existing or to be created objects~\cite{crowder2015multidisciplinary}. The main purpose of this modeling is to support reasoning about the structural behavior and properties of the objects. System architecture helps to consistently describe and efficiently design complex systems such as an industrial system or an organization which can be comprise of smaller parts called ``subsystems'' as shown in Fig~\ref{fig:system}-a. In the proposed method, each subsystem is considered as a metabolite. The approach is extensively described in Section 3.

\emph{Systems biology}, known as the system-level integration of experimental and computational studies in biology, plays a significant role in understanding complex network systems~\cite{kitano2002computational}.
In biological systems, functions emerge both from the elements and the network interconnections. This underscores the importance of computational analysis to extract useful information from biological data.

Reconstruction of genome-scale metabolic network models, which are among the best-studied models in biotechnology, is possible by collecting the gene-protein-reaction information from related biochemical databases, genome annotations and literature~\cite{henry2010high}.
Imposing constraints on the fluxes of a reconstructed biochemical network results in the definition of achievable cellular metabolic functions~\cite{price2004genome}.
Several mathematical representations of constraints are typically used for constraint-based modeling of metabolic networks, including \emph{flux-balance} constraints (e.g., conservation of metabolic fluxes) and \emph{flux bounds}. The former means the network should be at the steady-state condition and the latter limits the numerical ranges of network parameters such as the minimum and maximum range of fluxes for each reaction.

Elementary flux mode (EFM) analysis is a well-studied method in constraint-based modeling of metabolic networks. In EFM analysis, a network is decomposed into minimal functional pathways~\cite{klipp2008systems} based on the assumption of balanced metabolic fluxes.
All various analogous concepts for generating flux vectors of a network, like extreme pathways~\cite{papin2002extreme}) and minimal generating vectors~\cite{bordbar2014minimal}, produce a subset of EFMs.

%with applications in metabolic network analysis~\cite{lopar2014study}.
%%------------------------

The use of EFM analysis~\cite{schuster1994elementary,schuster2000general} is of interest in several biological applications.
Bioengineering \cite{schuster2002use},
phenotypic characterization \cite{radhakrishnan2010phenotypic},
drug target prediction \cite{parvatham2013drug} and
strain design \cite{machado2015co} can be mentioned among these applications.
%%------------------------

Double-description is a technique to enumerate all extreme rays of a polyhedral cone that is widely used in finding EFMs. Some of these methods use the stoichiometric matrix itself~\cite{schuster1994elementary} and some others use the null-space of the stoichiometric matrix to generate EFM candidates~\cite{wagner2004nullspace,urbanczik2005improved,quek2014depth}.
Computational approaches have been proposed to speed up double-description-based methods to compute EFMs~\cite{gagneur2004computation,terzer2008large} and some of them led to the development of computational tools such as \emph{Metatool}~\cite{von2006metatool} and \emph{EFMtool}~\cite{terzer2008large}.
Some recent methods, such as~\cite{ullah2016gefm}, try to bring the insights of graph-theory to generating EFM candidates. Hardware platforms such as GPUs are also used as accelerators to speed up the process of generating candidates~\cite{khalid2013heterogeneous}.
Besides, methods that explore a set of EFMs with specific properties, such as $K$-shortest EFMs~\cite{de2009computing}, or EFMs with a given set of target reactions~\cite{david2014computing}, have been proposed based on linear programming.

%%-------------------------

The main focus of this paper is to propose a system architecture based on the AND/OR graph data model~\cite{arabzadeh2017graph}.
In the AND/OR graph representation of a metabolic network, metabolites are considered as graph nodes and connected to each other through hyperarcs which model as reactions.
The input of this system is the stoichiometric matrix of the given metabolic network. Its output is the set of EFMs of the network as shown in Fig~\ref{fig:system}-b. The term \emph{systems biology} should not be confused with \emph{system architecture}. ``Systems biology'' is a general term which is used for describing the holistic view of a biological entity. In systems biology, finding EFMs is a classical problem which can help in understanding the global behavior of metabolism. Here, the problem of finding EFMs in the context of systems biology is modeled as a system architecture.
The smallest repetitive part of the system is considered as a module that emulates the function of a metabolite to find the set of minimal flux-balanced pathways known as EFMs.
The system has a network topology that enables a parallelizable computational scheme. Designing a model to map a biological network to a hardware platform to take advantage of the potential multicore computational power, and not necessarily just a hardware accelerator, is the main contribution of this paper.
The underlying innovation in our proposed method is not in a parallel computing implementation of the existing algorithms. Instead, we consider a metabolic network as a system with seemingly independent subsystems that have very intricate relationships with their surroundings. While each subsystem is acting in parallel with the other ones, it is tightly coupled with the rest in terms of the inputs it receives and outputs it generates. We call this architectural or topology-based parallelism. Inspired by this observation, we define a model that closely resembles a real metabolic network. In this model each metabolite is represented by an ``independent'' processing element that is busy performing its own reactions by processing its inputs and producing its outputs. Once this model is established, GPU seems like a natural choice to implement the notion of topology-based parallelism.

%EFMs are a subset of this set. Assuming that there is no delay in the network when fluxes get values, the semi-minimal set is minimal and exactly equals the EFMs. Since the goal is to find EFMs, in the rest of the paper the term EFM is used for the candidate pathways generated by the system.

%...........................................
\begin{figure}[!h]
\centering
\includegraphics[scale=0.8]{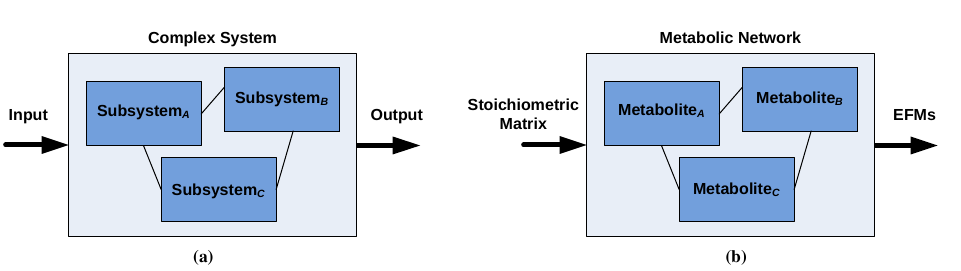}
\caption{{\bf System architecture.}
(a) Designing a distributed and analyzable structure for a complex system consider as a system architecture. (b) System architecture view of the proposed method.}
\label{fig:system}
\end{figure}
%...........................................

%%-------------------------
%%The rest of the paper is organized as follows. Preliminaries are provided in Section \emph{Preliminaries}. The main system design along with required definitions are discussed in Section \emph{Materials and methods}. Section \emph{Implementation and analysis} is devoted to implementation and results. In \emph{Technical details} and \emph{Discussion} sections design metrics and specifications are discussed and, finally, Section \emph{Conclusion} concludes the paper.

The rest of the paper is organized as follows. Preliminaries are provided in Section 2. The main system design along with required definitions are discussed in Section 3. Section 4 is devoted to implementation and results. The design metrics and specifications are discussed in Sections 5 and 6, finally, Section 7 concludes the paper.

%\IEEEpubidadjcol
%%-----------------------------------------------------------------------------------------------------------------------------------
%%-----------------------------------------------------------------------------------------------------------------------------------
\section{Preliminaries}\label{sec:premi}
In this section, some basic concepts of metabolic networks and their counterpart graph data model, formal definition of EFMs and a brief introduction to parallel computing are provided.
\subsection{Metabolic networks}
Metabolic networks model the metabolism of living cells in terms of a set of biochemical reactions. The biochemical reactions can be irreversible or reversible which means the reaction can be active only in one direction, or can be active in both directions, respectively. The contributing metabolites in a reaction can be either substrates or products. Substrates are consumed and products are produced during the operation of a reaction.
The topology of a metabolic network is characterized by its $m \times n$ stoichiometric matrix, \textbf{S}, where $m$ and $n$ correspond to the number of metabolites and the number of reactions, respectively. The value $S_{ij}$ represents the stoichiometric coefficient of the metabolite $i$ in the reaction $j$. $S_{ij}$ is positive/negative if the metabolite $i$ is produced/consumed. If this coefficient is zero it means that the metabolite $i$ does not contribute to the reaction $j$.
The network is considered in the steady-state if for each internal metabolite, the rates of consumption and production are equal. The reactions connected to the external metabolites are called \emph{Boundary} reactions.
%Definitions in this section are derived from~\cite{zanghellini2013elementary}.\\

%%-------------------------------------------------------------
\noindent
\textbf{Definition 1. Flux Mode.}
A \emph{flux mode} $\textbf{v} \in \mathcal{R}^n$ illustrates flux distributions of a set of reactions in a given metabolic network. Non-zero values in $\textbf{v}$ represent reaction fluxes.

%%-------------------------------------------------------------
\noindent
\textbf{Definition 2. Flux-Balanced Mode.}
A flux mode $\textbf{v} \in \mathcal{R}^n$ and $\textbf{v} \neq \textbf{0}$ is called \emph{flux-balance}, if it meets the following conditions:
\begin{itemize}
  \item $v_{i} \geq 0$ for all $i \in$ irreversible reactions (thermodynamic constraint) and
  \item $\textbf{S}.\textbf{v} = \textbf{0}$, i.e., the rates of consumption and production of internal metabolites are equal (steady-state condition). The low-dot operator simply is the matrix inner product.
\end{itemize}

%%-------------------------------------------------------------
\noindent
\textbf{Definition 3. Elementary Flux Mode.}
A flux mode $\textbf{v} \in \mathcal{R}^n$ and $\textbf{v} \neq \textbf{0}$ is considered an EFM, if it meets the following conditions:
\begin{itemize}
  \item $\textbf{v}$ is flux-balanced based on Definition 2,
  \item there is no $\textbf{v}^{\prime} \in \mathcal{R}^n$ with $\textbf{supp}(\textbf{v}^{\prime}) \subset \textbf{supp}(\textbf{v})$, where support of a mode is defined as $\textbf{supp}(\textbf{v})= \{i|v_i \neq 0\}$, (minimality limitation).
\end{itemize}
%%-------------------------------------------------------------

In Fig~\ref{fig:exNet}, an example of a metabolic network is illustrated.

%...........................................
\begin{figure}[!h]
\centering
\includegraphics[scale=0.7]{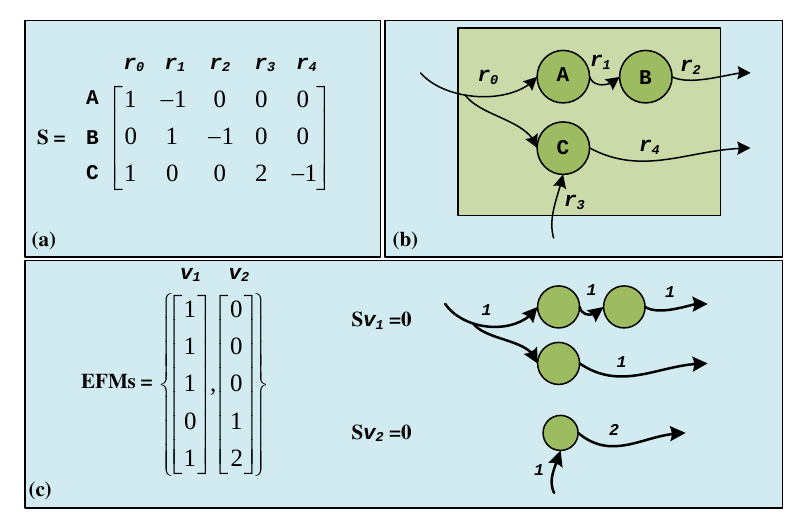}
\caption{{\bf An example of a metabolic network.}
(a) Stoichiometric matrix of the network. (b) Graph representation of the network. (c) The set of EFMs of the network.}
\label{fig:exNet}
\end{figure}
%...........................................

%%-------------------------------------------------------------
\subsection{Metabolic network data model}
The used data model in this paper is derived from~\cite{arabzadeh2017graph}. The model is based on the conventional AND/OR graph in computer science with additional features to make the model appropriate for object-oriented methods considering the metabolites as objects. In this model, the coefficients of the metabolites in reactions are embedded in the graph structure as attributes of each node (defined in Definition 6). The information of the candidate pathways is also embedded in the graph structure.\\

%%--------------------------------------------------------------------------------------------------------
\noindent
\textbf{Definition 4.} Modified graph for representing a metabolic network, denoted as $\mathbf{MG}$, is defined as a set of \emph{Nodes}, i.e., $\mathbf{MG}$ = $\{\mathcal{N}_i|0\leq i\leq {M-1}\}$, where $M$ is the number of internal metabolites and each $\mathcal{N}_i$ represents a metabolite in the network.\\

%%--------------------------------------------------------------------------------------------------------
\noindent
\textbf{Definition 5.} Each node $\mathcal{N}_i$ in $\mathbf{MG}$ is a 3-tuple $\mathcal{N}_i$ = ($i$, $I$, $O$) where
\begin{itemize}
  \item $i$ is the tag of a metabolite,
  \item $I$ is an array of input reactions that produce the metabolite and
  \item $O$ is an array of output reactions that consume the metabolite.
\end{itemize}

%%--------------------------------------------------------------------------------------------------------
\noindent
\textbf{Definition 6.} Each $I$/$O$ in Definition 5 contains the following data:
\begin{itemize}
  \item The reaction $j$, $0\leq j\leq r-1$, where $r$ is the number of reactions consuming/producing the metabolite $i$,
  \item $I_{M_j}$/$O_{M_j}$, an array of the metabolites consumed/produced by reaction $j$. In other words, $I_{M_j}$/$O_{M_j}$=$\{m_{kj}|0\leq k\leq m-1\}$, where $m$ is the number of consumed/produced metabolites by the reaction $j$,
  \item $I_{{\mathord{\buildrel{\lower3pt\hbox{$\scriptscriptstyle\frown$}}\over M}}_j}$/$O_{{\mathord{\buildrel{\lower3pt\hbox{$\scriptscriptstyle\frown$}}\over M}}_j}$, an array of the metabolites produced/consumed by reaction $j$. In other words, $I_{{\mathord{\buildrel{\lower3pt\hbox{$\scriptscriptstyle\frown$}}\over M}}_j}$/$O_{{\mathord{\buildrel{\lower3pt\hbox{$\scriptscriptstyle\frown$}}\over M}}_j}$=$\{m_{kj}|0\leq k\leq m-1\}$, where $m$ is the number of produced/consumed metabolites by the reaction $j$ excluding the metabolite $i$ itself,
  \item The direction of the reaction for reversible reactions,
  \item $I_{c_{ij}}$/$O_{c_{ij}}$, the coefficient of the reaction $j$ for the produced/consumed metabolite $i$ in the stoichiometric matrix $S$ and
  \item $I_{f_{ijp}}$/$O_{f_{ijp}}$, the flux of the input/output reaction $j$ for a certain path $p$.
\end{itemize}

\noindent
For bidirectional reactions, the term ``direction'' is used to clarify that either the reactants react to form the products, or the products react together to produce the reactants in the backward direction. A \emph{hyperarc} links a set of nodes to another set in one connection. The incoming hyperarcs (i.e., $I$ in Definition 5) in each $\mathcal{N}_i$ produce the metabolite $i$ and the outgoing hyperarcs (i.e., $O$ in Definition 5) consume it. Therefore, consumed metabolites, $m_i$, and produced metabolites, $m^{\prime}_{i}$, contributing to reaction $j$ as shown in Eq~\ref{eq:1}, are as hyperarcs between $\mathcal{N}_i$ nodes each associated to one metabolite. $I$ hyperarcs and $O$ hyperarcs in $\mathcal{N}_i$ nodes are related to each other by reaction tags in Definition 6. The set of metabolites in $m_i$, $m^{\prime}_{i}$, are \emph{AND-related}. The incoming $I$ (outgoing $O$) hyperarcs for each $\mathcal{N}_i$ are \emph{OR-related}.

\begin{equation}\label{eq:1}
\small
r_j: m_1 + m_2 + ... + m_i \rightleftharpoons m^{\prime}_{1} + m^{\prime}_{2} + ... +m^{\prime}_{i}\textbf{.}
\end{equation}
%...........................................
%%-------------------------------------------------------------
\noindent
\textbf{Definition 7. Pathway.}
A pathway in a graph data model is defined as a chain of adjacent $\mathcal{N}_i$ nodes such that they form a hyperpath from a source node to a sink node. Adjacent nodes represent the metabolites which contribute in a common reaction.

%%-------------------------------------------------------------
A pathway can be seen as a subgraph of the AND/OR representation of a metabolic network. In addition, a pathway can be represented as a vector in a convex flux cone. It should be noted here that these two definitions are shown to be equivalent~\cite{ zevedei2003topological}.

%%-------------------------------------------------------------
%%-------------------------------------------------------------
\subsection{Parallel computing and programming}
Parallel computing is a type of computation in which independent executable elements of a system can run simultaneously. This is based on the assumption that large problems can often be divided into smaller ones, which can then be solved at the same time.
Solving complex scientific and engineering problems require huge computational power.
Processing power in high performance computing (HPC) is measured by floating point operations per second, or \emph{FLOP/s} and the size of data in \emph{Bytes} considering that a double precision floating point number takes 8 bytes of memory.
Memory and timing constraints are among principal reasons that move the trend to parallel computing and programming frameworks.
Each memory access can take several CPU cycles and increasing clock frequency does not readily improve performance anymore.
%In addition, increasing frequency and decreasing transistor size lead to both static and dynamic power issues.

There are several forms of parallel computing such as instruction-level and thread-level schemes.
Instruction-level parallelism (ILP) is defined as a parallelism among individual instructions executed by a microprocessor. Parallel parts are automatically extracted by the microprocessor.
This type of parallelism is limited by the pipeline depth and instruction dependencies.
%A thread is a stream of instructions with execution context and
On the other hand, thread-level parallelism (TLP) happens when a set of multiple concurrent tasks are running and the programmer\textquotesingle s involvement is required to extract the parallelism as done in our model.

Considering data along with instructions, two models are defined. In single instruction multiple data (SIMD) model, a large number of (usually small) processing elements apply a single instruction on multiple data sets. In other words, a single controlling processor issues each instruction and each processing element executes the same instruction.
Multiple instruction multiple data (MIMD) is another model in which each processor executes its own sequence of instructions independently.
An extended model of SIMD, single thread multiple data (STMD), is used in the proposed model and will be explained later.

A programming model is a bridge between a system developer\textquotesingle s natural model of an application and an implementation of that application on the target hardware. A programming model must allow the programmer to balance the competing goals of productivity, in terms of time and resource, and implementation efficiency~\cite{asanovic2006landscape}.

Different hardware platforms and programming environments are introduced in the recent years to address the parallel computing requirements. The proposed system independent of the underlying hardware is presented in Section 3 and the appropriate hardware for our proposed model and its programming requirements are discussed in Section 4.

%%-----------------------------------------------------------------------------------------------------------------------------------
%%-----------------------------------------------------------------------------------------------------------------------------------
\section{Materials and methods}\label{sec:method}
In this section, the proposed system architecture is introduced and then the relationship between the system and the applied data model is discussed.
The main target of the system is to model the function of each metabolite and create minimal pathways along these functional systems.
To better describe the model, a set of terms are defined in the following section.

\subsection{Definitions}
\noindent
\textbf{Definition 8. Pathway Creation.}
\emph{Pathway creation} in a graph data model is defined as the process of starting from a boundary reaction (i.e., an external arc) and following the chain of reaction-metabolites, recursively adding the AND-related nodes to a list. When a reaction adds a metabolite to the list of the pathway, the metabolite is called to be triggered by the reaction.

\noindent
\textbf{Definition 9. Forward/Backward Flow.}
The term $forward$ is applied to a flow when the process of pathway creation is in a forward direction for a particular metabolite, which means that the metabolite is produced according to the flux changes of the input/output reactions of that metabolite. The term $backward$ is applied when the flow is in a backward direction, i.e, the metabolite is consumed according to the flux changes.\\

\noindent
\textbf{Definition 10. Primary and Secondary Reactions.}
On each pathway, each node (metabolite) has a \emph{primary} input (PI) and a \emph{primary} output (PO). In forward/backward flow, a PI/PO is the first reaction which enters a node and a PO/PI is the first reaction which directs the flow of the pathway and tags the node as \emph{visited}. When an edge (i.e., a reaction) enters an already visited node in that pathway in a forward/backward flow as an input/output, that reaction is tagged as a \emph{secondary} input/output (SI/SO) of that node.\\

For example, when metabolites $A$ and $B$ are AND-related, if $A$ is consumed during the reaction $A$+$B$$\rightarrow$$C$+$D$, then $B$ should also be present in the medium which means $A$ and $B$ should be consumed simultaneously. The forward/backward flow is designed for the produced/consumed metabolites~\cite{arabzadeh2017graph}. In this example, the metabolites $A$ and $B$ are consumed while $C$ and $D$ are produced. The order of adding nodes to the list in the pathway creation process in the forward/backward flow names reactions as primary or secondary. An example of primary and secondary reactions is illustrated in the graph model of Fig~\ref{fig:PriSec}.

%...........................................
\begin{figure}[!h]
\centering
\includegraphics[scale=0.9]{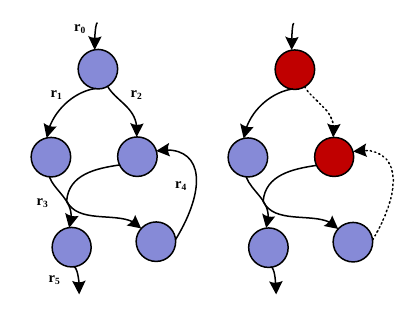}
\caption{{\bf Primary and secondary reactions.}
The pairs (r0,r1), (r1,r3), (r3,r5), (r3,r2), (r3,r4) act as primary reactions (input, output) for the
five nodes since they tag the nodes as visited in the pathway creation process. The two nodes with dashed reactions are already visited when r2 and r4 enter, so the reactions are tagged as secondary reactions.}
\label{fig:PriSec}
\end{figure}
%...........................................

\noindent
\textbf{Definition 11. Flux-dependent reactions.}
A primary output and a secondary input in a forward flow (or a primary input and a secondary output in a backward flow) are called flux-dependent reactions if the pathway which goes through the primary reaction, extends to the secondary reaction. In this case, their fluxes are linearly dependent over the target node. The dependency is shown by the sign $\propto$ in the rest of the paper. The two dependent types are illustrated in Fig~\ref{fig:FluxDep}.\\

%...........................................
\begin{figure}[!h]
\centering
\includegraphics[scale=0.55]{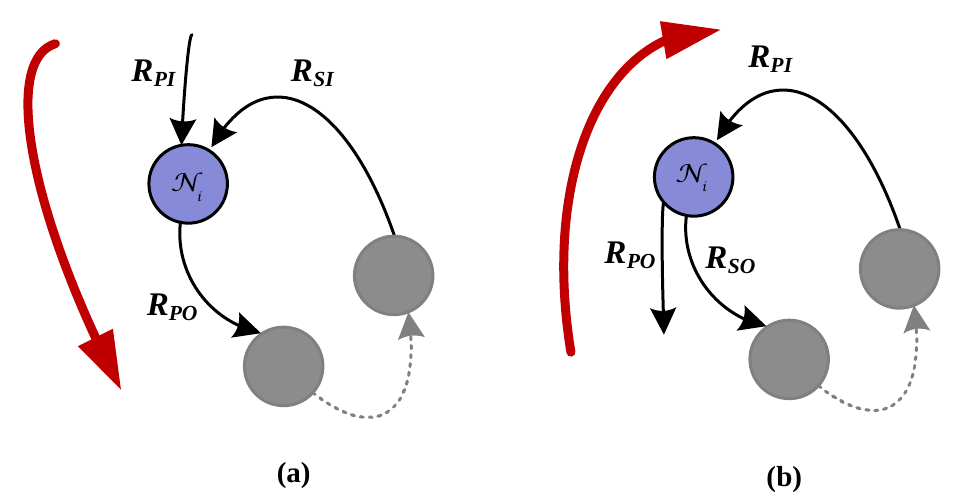}
\caption{{\bf Flux dependencies.}
Reactions with flux dependencies with respect to metabolite $\mathcal{N}_i$. (a) $R_{PO}$ (primary output) and $R_{SI}$ (secondary input) are recognized as flux-dependent reactions in a forward flow, $R_{PO} \propto R_{SI}$. (b) $R_{PI}$ (primary input) and $R_{SO}$ (secondary output) are recognized as flux-dependent reactions in a backward flow, $R_{PI} \propto R_{SO}$.}
\label{fig:FluxDep}
\end{figure}

\subsection{The proposed architecture}
In this section, the main architecture of the system to model and analyze metabolic networks is provided. The main goal of the system is to create minimal  flux-balanced metabolic pathways. In the proposed architecture, each node in an $\mathbf{MG}$ graph is considered as a processing element or module with a local memory. The whole system has a global memory used by the processing elements to synchronize their tasks.
The concept behind the model is that by starting from boundary metabolites (or any other internal metabolite set), each incoming flux in a pathway produces a set of metabolite(s), and each produced metabolite should be consumed. The flux of the contributed reactions can be obtained according to the stoichiometric coefficients in order to make sure that the constructed pathway is in steady-state. Only one output from the node should be considered when constructing a pathway. When one metabolite is consumed by a reaction, the presence of its AND-related metabolites is required. Therefore, the algorithm is designed to have a forward/backward flow for the produced/consumed metabolites.
The system architecture is illustrated in Fig~\ref{fig:design}. It consists of $m$ META modules, each corresponding to one node, and an arbiter which evaluates the constructed pathways and decides whether they belong to a subset of the desired pathways or not.
The modules are introduced in further detail below.

%...........................................
\begin{figure}[!h]
\centering
\includegraphics[scale=0.7]{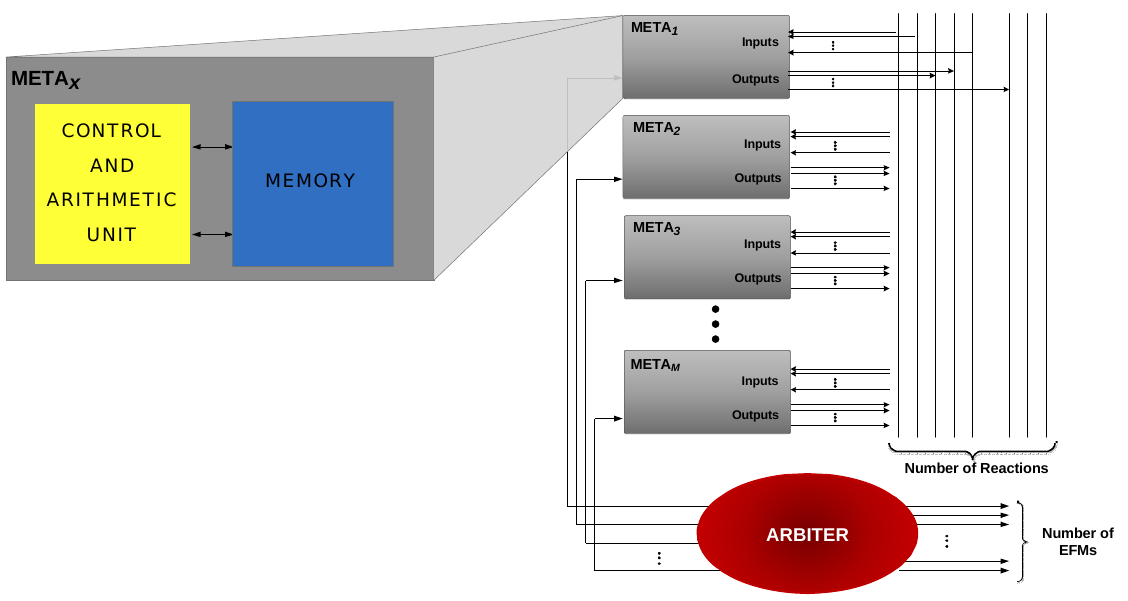}
\caption{{\bf Main architecture.}
The main architecture of the hardware model of a metabolic network.}
\label{fig:design}
\end{figure}
%...........................................
%\begin{figure*}[!tpb]
%\centering
%\includegraphics[scale=0.7]{files/design}
%\caption{The main architecture of the hardware model of a metabolic network.}
%\label{fig:design}
%\end{figure*}
%...........................................

\subsubsection{Modules: META$x$}
Each META module with the index $x$ consists of a control/arithmetic unit and a memory unit. The task of the module is to keep the information of the pathways (i.e., EFM candidates) and to execute the required instructions based on the obtained information from the neighboring META modules. Two modules are neighbors if they are data dependent (i.e., have shared reactions). The two internal units of the META$x$ module are discussed below.

\textbf{Control and arithmetic unit.}
Fig~\ref{fig:controlChart} illustrates the flowchart of the operation of the control and arithmetic unit. The task of all boxes labeled as PROC in the figure are summarized in Table~\ref{tab:table1}.

%...........................................
\begin{figure}[!h]
\centering
\includegraphics[scale=0.55]{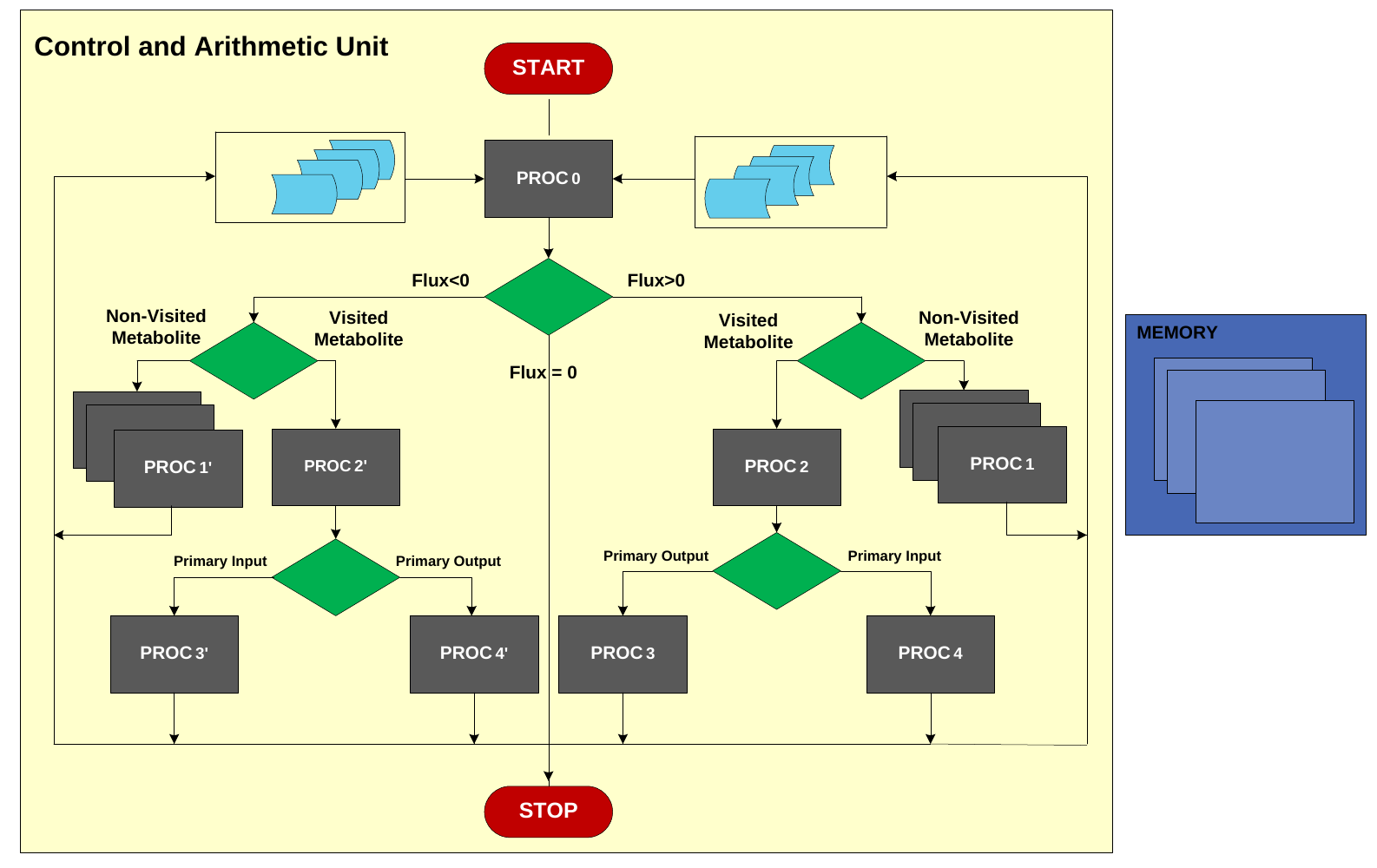}
\caption{{\bf Control and arithmetic unit.}
The flowchart of the control and arithmetic unit of the META$x$ module.}
\label{fig:controlChart}
\end{figure}
%...........................................
%\begin{figure*}[!tpb]
%\centering
%\includegraphics[scale=0.5]{files/controlChart}
%\caption{The flowchart of the control and arithmetic unit of the META$x$ module.}
%\label{fig:controlChart}
%\end{figure*}
%...........................................
%\input{files/table1.tex}
\begin{table}%[!tpb]
\caption{The task of the processes in Fig~\ref{fig:controlChart}.}
\label{tab:table1} {
\centerline{
\scriptsize
%\small
%\tiny
\begin{tabular}{ll}
\hline
Process Name		&Task  													\\
\hline
\hline
PROC$_{0}$		 &Calculate metabolite balance for each pathway using Eq. (2)		 				\\
PROC$_{1(1^\prime)}$	&For each output/input, start a new pathway using Eq. (3) and go to PROC${_0}$; 	\\
                        &Tag $Primary$ input/output reactions                             \\
PROC$_{2(2^\prime)}$	&Check reaction dependency status		 							\\
PROC$_{3(3^\prime)}$	&Use Eq. (4) to balance the flux	and go to PROC${_0}$ $or$  	 	\\
                        &go to STOP if MAX\_LOOP is reached                                                                   \\
PROC$_{4(4^\prime)}$	&Use Eq. (5) to balance the flux	and go to PROC${_0}$ $or$  		\\
                        &go to STOP if MAX\_LOOP is reached                                                   \\
\hline
\end{tabular}}
}
\end{table}
%...................................................

The thread inside the unit is running continuously until the arbiter closes the processing element or the maximum times a node is visited, represented by MAX\_LOOP, is reached.
In the beginning, the metabolite balance of each pathway, is calculated in PROC$_{0}$ based on Eq~\ref{eq:2}.

\begin{equation}\label{eq:2}
EF = \sum {I_{c_{ij} }I_{f_{ijp} }  }  - \sum {O_{c_{ij} }O_{f_{ijp} }  }.
\end{equation}

In this equation, \emph{EF} indicates the extra flux. Depending on the value of \emph{EF}, one of the three routes in the chart should be selected. If the metabolite is visited for the first time, new pathways are created for each output/input in PROC$_{1}$/PROC$_{1^\prime}$ using Eq~\ref{eq:3}.

\begin{equation}\label{eq:3}
\begin{array}{l}
 O_{f_{ijp} }  = {\rm   }\frac{{I_{f_{ijp} } I_{c_{ij} } }}{{O_{c_{ij} } }},\,\,\,\,\,\,\,\,\,\,\,\,\,\,\,\,EF>0\,\,\,\,\,\,\, \rm{(Forward)}\\\\
 I_{f_{ijp} }  = {\rm   }\frac{{O_{f_{ijp} } O_{c_{ij} } }}{{I_{c_{ij} } }},\,\,\,\,\,\,\,\,\,\,\,\,\,\,\,EF<0\,\,\,\,\,\,\,\rm{(Backward)}.\\
 \end{array}
\end{equation}

Otherwise in PROC$_{2}$/PROC$_{2^\prime}$ the reaction dependency is checked between the input/output flux-changed reactions. Based on the results of PROC$_{2}$/PROC$_{2^\prime}$, if reactions are independent, Eq~\ref{eq:4} is used in PROC$_{3}$/PROC$_{3^\prime}$ to pass the flux to PO/PI reactions.

\begin{equation}\label{eq:4}
\begin{array}{l}
 I_{f(n)_{ijp} }  = \frac{{I_{f(o)_{ijp} } I_{c_{ij} }  + O_{f_{iep} } I_{c_{ie} } }}{{I_{c_{ij} } }},\, \\\\
 O_{f(n)_{ijp} }  = \frac{{O_{f(o)_{ijp} } O_{c_{ij} }  + O_{f_{iep} } I_{c_{ie} } }}{{O_{c_{ij} } }}. \\\\
 \end{array}
\end{equation}

Otherwise, Eq~\ref{eq:5} is used and fluxes are passed to PI/PO reactions. In both PROC$_{3}$/PROC$_{3^\prime}$ and PROC$_{4}$/PROC$_{4^\prime}$, a variable which keeps track of the number of times a metabolite is visited is checked against a constant MAX\_LOOP to avoid the loop over a node forever.

\begin{equation}\label{eq:5}
\begin{array}{l}
 I_{f(n)_{ijp} }  = \frac{{I_{f(o)_{ijp} } I_{c_{ij} }  - I_{f_{iep} } I_{c_{ie} } }}{{I_{c_{ij} } }},\, \\\\
 O_{f(n)_{ijp} }  = \frac{{O_{f(o)_{ijp} } O_{c_{ij} }  - I_{f_{iep} } I_{c_{ie} } }}{{O_{c_{ij} } }}. \\
 \end{array}
\end{equation}

\textbf{Memory unit.}
The data stored in the memory unit of the META$x$ can be divided into two categories.
\begin{itemize}
  \item \emph{Metabolite local}: Consists of static and dynamic information. Static information of each metabolite includes input/output stoichiometry coefficients and dynamic information includes the metabolite status in each pathway which is changed over time.
  \item \emph{Path local}: Consists of dynamic information of a pathway.
\end{itemize}
Static and dynamic terms indicate Read-Only and Read/Write memory types, respectively. Fig~\ref{fig:memory} illustrates the structure of the memory unit.

%...........................................
\begin{figure}[!h]
\centering
\includegraphics[scale=0.6]{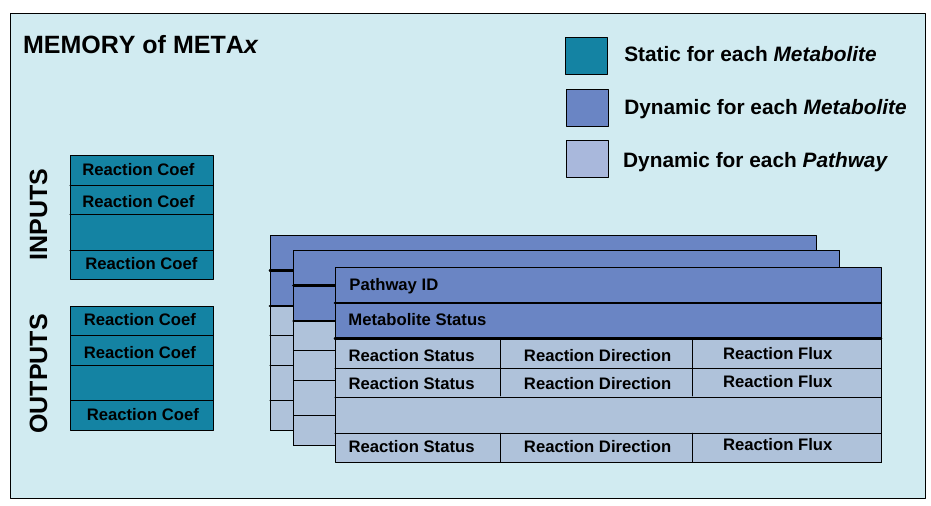}
\caption{{\bf Memory unit.}
The structure of the memory unit of META$x$..}
\label{fig:memory}
\end{figure}
%...........................................
%\begin{figure*}[!tpb]
%\centering
%\includegraphics[scale=0.55]{files/memory}
%\caption{The structure of the memory unit of META$x$.}
%\label{fig:memory}
%\end{figure*}
%...........................................

\subsubsection{Modules: Arbiter}
Based on the pathway information stored in META$x$ modules, the arbiter decides which pathway candidates are balanced and to be considered as EFMs. When all META$x$ modules are done with the analysis of all in-process pathways, the arbiter reports the result. The function of the arbiter is considered as a part of META$x$ function in Algorithm 3. The decision is made based on the status of each pathway specified in the algorithm as discussed later.

%\subsubsection{Connections}
%Connections between modules are depicted in Figure~\ref{fig:design} as wires to show how META$x$ modules and arbiter are connected through reaction statuses. This can be stored in local memory of each module accessible by the others or a global memory.

%%-----------------------------------------------------------------------------------------------------------------------------------
%%-----------------------------------------------------------------------------------------------------------------------------------
\section{Implementation and analysis}\label{sec:IMPRES}
The hardware platform introduced in Section 3 is implemented on a graphic processor unit (GPU). To demonstrate the operation of this platform, a metabolic network is considered for experiments and the results are provided in the remainder of this section. The platform consists of a set of parallel threads of STMD type where each META$x$ module runs the same thread for various data sets. Complication arises because of the inevitable dependencies between the modules\textquotesingle s data.

\subsection{General-purpose graphics processor units}
GPUs were originally developed for computationally intensive graphics of games and image processing applications in computer systems. But recently, their vast computational power has been utilized in general purpose computing.
%General-purpose computing on graphic processor units (or GPGPU) is a use of GPUs for general-purpose applications through CPU.

There are several reasons to consider GPUs as the hardware platform to emulate metabolic networks in this research.
Firstly, a suitable environment for software development based on C/C++ is available for high-level and easy programming which is convenient to try alternative options for running test routines.
Secondly, a large number of floating point units (FPUs) in GPUs can operate in parallel to accelerate our analysis.

A general architecture of a GPU is shown in Fig~\ref{fig:gpu}-a in which the main program is running on a CPU. To run a process on the GPU, first the data in the CPU memory is copied to the GPU memory. When the process is done on the GPU, the result is copied from the GPU memory to the CPU memory. The parallel functional units of the GPU are $blocks$, each consisting of a set of $threads$. All threads in all blocks perform the same function in one access to the GPU. In other words, the main stream of the program is running on the CPU and function calls from the CPU initialize a process on all GPU threads. By the term \emph{access}, we mean the function call that is invoked by the CPU to setup the GPU. Therefore, all threads run simultaneously from a programmer\textquotesingle s point of view. A typical GPU has three memory levels: a local memory for each thread, a shared memory between threads of a block and a global memory.

The correspondence between the system architecture and the GPU platform is illustrated in Fig~\ref{fig:gpu}-b. Each META$x$ element stands on a thread and each block is dedicated to a candidate EFM pathway which means that different pathways are processed by different blocks.

%...........................................
\begin{figure}[!h]
\centering
\includegraphics[scale=0.7]{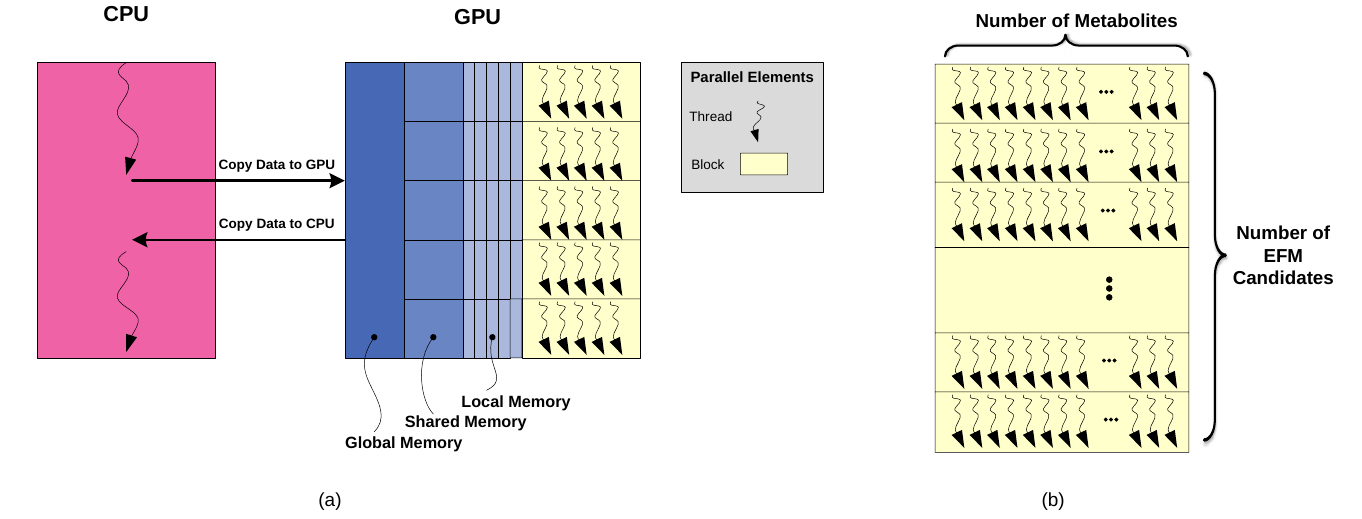}
\caption{{\bf The proposed system architecture on the GPU platform.}
(a) GPU general architecture and the way GPU and CPU communicate. (b) How the proposed system architecture is mapped onto the GPU.}
\label{fig:gpu}
\end{figure}
%...........................................
%\begin{figure*}[!tpb]
%\centering
%\includegraphics[scale=0.7]{files/gpu}
%\caption{(a) GPU general architecture and the way GPU and CPU communicate. (b) How the proposed system architecture is mapped to the GPU.}
%\label{fig:gpu}
%\end{figure*}
%...........................................

Algorithm~1 shows the main function running on the CPU. Three accesses, called \emph{kernel} functions, are made from CPU to GPU as shown in Fig~\ref{fig:algFig}. Details of each function are discussed below.

%...........................................
%\input{files/algM.tex}
\begin{algorithm}[!tpb]
\label{alg:main}
\caption{Main body of the algorithm.}
\textbf{         }\\
\DontPrintSemicolon
\SetAlgoLined
\scriptsize
\KwIn{The stoichiometric matrix of a given metabolic network.}
\KwOut{A set of minimal flux-balanced pathways of the network.}
\textbf{         }\\
\#define Number of input metabolites as     $M$		\\
\#define Number of input reactions as       $R$		\\
\#define Number of EFM Candidates as        $C$	    \\
\#define The selected depth value to find dependent reactions as     $L$     \\
\textbf{         }\\
\#define Max number of input/output reactions of a metabolite                $r_{max}$  \\
\#define Max number of input/output metabolites of a reaction                $m_{max}$  \\
\textbf{         }\\
\#define Number of blocks and threads pair as ($\#B$,$\#T$)           \\
\textbf{         }\\
Convert the given matrix to the $\mathbf{MG}$ graph structure.\\
Copy $\mathbf{MG}$ structure from CPU to GPU global memory.\\
\textbf{         }\\
\tcc{Finding dependent reactions with depth $L$ over all metabolites.}
Set ($\#B$,$\#T$) as ($M$,${r_{max}}^{L}$) to be run on GPU.\\
Run DEPTHy() on each thread.\\
\textbf{         }\\
\tcc{Initializing selected reactions with a nonzero value in all $Blocks$.}
Set ($\#B$,$\#T$) as ($C$,$R$) to be run on GPU.\\
Run INIT() on each thread.\\
\textbf{         }\\
\tcc{Calculating balanced pathways by keeping all metabolites in steady-state.}
Set ($\#B$,$\#T$) as ($C$,$M$) to be run on GPU.\\
Run METAx() on each thread.\\
\textbf{         }\\
Copy $\mathbf{MG}$ pathways information from GPU to CPU.\\
%\tcp*{}
%\tcc{}
\end{algorithm}
%...........................................
\begin{figure}[!h]
\centering
\includegraphics[scale=0.6]{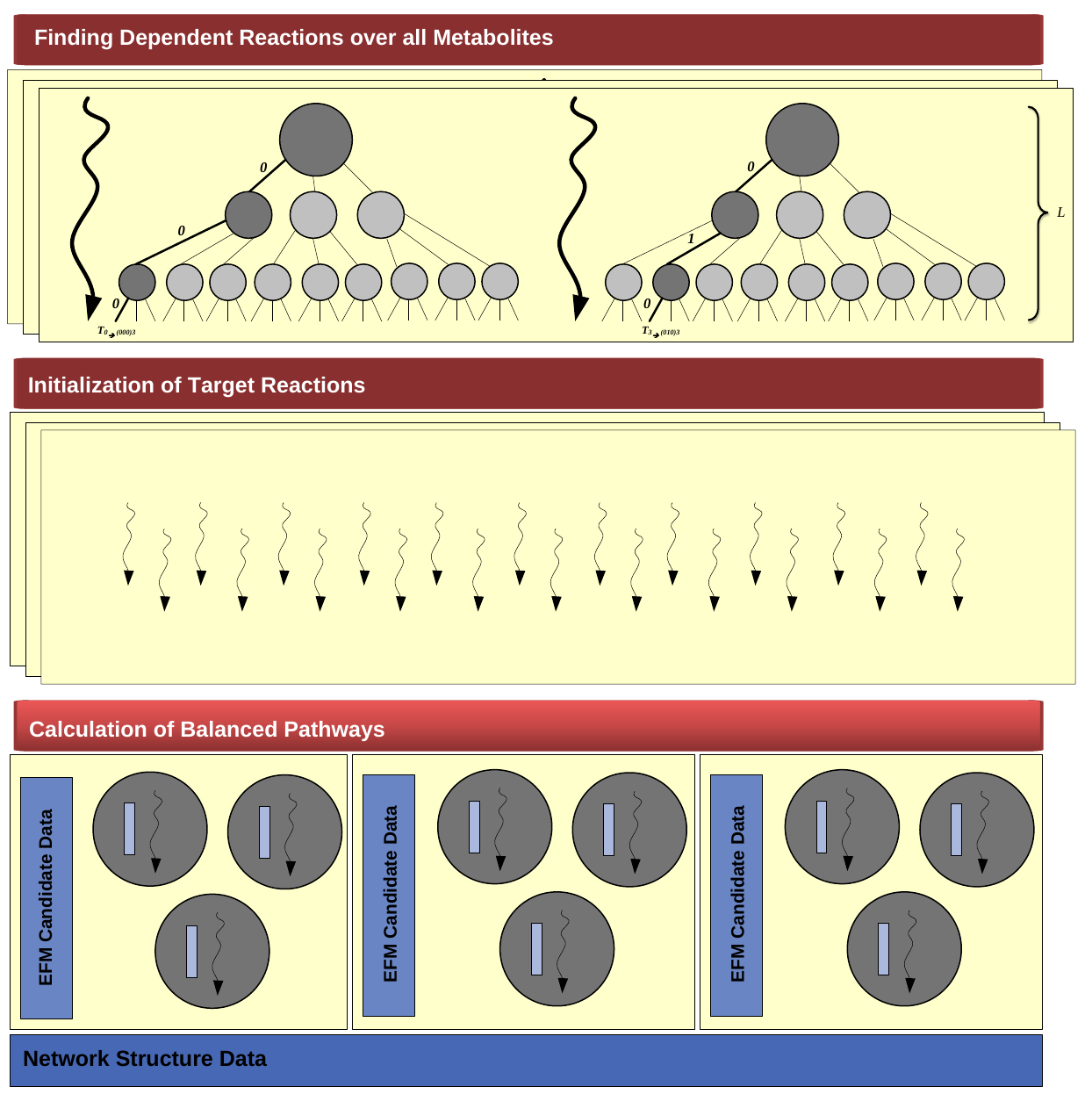}
\caption{{\bf Kernel functions.}
Three kernel functions are illustrated in the figure. First, dependent reactions are explored. Each pathway is assigned to one thread and the thread number in base $r_{max}$ in $L$ digits shows the order of the reactions. Two example for thread IDs 0 and 3 are shown in the figure. $r_{max}$ and $L$ are both considered as three here. The conversion of 0 and 3 to the base $r_{max}$ gives us $(000)_3$ and $(010)_3$ which specify the route of the pathway in the assumed thread. Each yellow box is considered as a block and the process repeats independently for all metabolites each of which dedicated to one block.
After initializing the target reactions in one pass, each block is assigned to an EFM candidate and each thread performs the META$x$ function. The blue rectangular shapes are memory elements showing the hierarchical memory levels in GPU.}
\label{fig:algFig}
\end{figure}
%...........................................
%\begin{figure*}[!tpb]
%\centering
%\includegraphics[scale=0.6]{files/algFig}
%\caption{
%Three kernel functions are illustrated in the figure. First, dependent reactions are explored. Each pathway is assigned to one thread and the thread number in base $r_{max}$ in $L$ digits shows the order of the reactions. Two example for thread IDs 0 and 3 are shown in the figure. $r_{max}$ and $L$ are both considered as three here. The conversion of 0 and 3 to the base $r_{max}$ gives us $(000)_3$ and $(010)_3$ which specify the route of the pathway in the assumed thread.
%Each yellow square is considered as a block and the process repeats independently for all metabolites each of which dedicated to one block.
%After initializing the target reactions in one pass, each block is assigned to an EFM
%candidate and each thread performs the META$x$ function. The blue rectangular shapes are memory elements showing the hierarchical memory levels in GPU.}
%\label{fig:algFig}
%\end{figure*}
%...........................................
%...........................................

\textbf{Finding dependent reactions over all metabolites.}
In the flowchart of the control/arithmetic unit (Fig~\ref{fig:controlChart}), flux-dependent reactions of a metabolite should be identified to make better decisions.
Algorithm~2 describes the DEPTH$y$ function which starts a path through a metabolite $\mathcal{N}_i$ and continues until the path length (i.e., the number of involved reactions on the path) is equal to $L$. If $\mathcal{N}_i$ is reached again in the created path with length $\leq L$, the involved reactions of $\mathcal{N}_i$ are characterized as flux-dependent.
To do this using the GPU, one block is dedicated to each metabolite. Each thread of each block is dedicated to a path starting from $\mathcal{N}_i$. Therefore, DEPTH$y$ function is executed on each thread of each block to traverse paths and resolve dependencies. Each thread of a block is identified by a block ID and a thread ID. The block ID is set to $i$ for the metabolite $\mathcal{N}_i$. The thread ID should specify a unique path. To do this, the thread ID is used to route the path and to reveal the route. We propose a static path-generator method in which the thread ID $j$ as a decimal number is converted to a number in base $r_{max}$ and length $L$ as $j$=$(L_{0\leq l < L})_{r_{max}}$=$(L-1, L-2,...,l,...,1,0)_{r_{max}}$.
$r_{max}$ is the maximum number of output reactions of a node $\mathcal{N}_i$. So, the pair ($\#$ of blocks, $\#$ of threads) is set as ($M$,${r_{max}}^{L}$) and DEPTH$y$ is set up on each.
We start from a node $\mathcal{N}_i$. The index $l$ counts the number of nodes traversed on the graph, that is, the parameter $depth$.
%Besides, the $l$-th digit of $(L)_{r_{max}}$ is the order of the reaction which should be traversed in the $l$-th node in the queue $Q$ shown as $Q_l$.
The digit in position $l$ represents the order of the reaction which should be selected in depth $l$. Consider that the node in depth $l$ in the queue $Q$ is represented by $Q_l$. The number of outputs of $Q_l$ is considered as $r$ which is the size of array $O$ of the node. If $l<r$, the process goes on, otherwise the path is undefined.
To continue the process, check to see if the node $Q_l$ is in the set of output metabolites of reaction $L_l$ shown by $O_{M_{L_l}}$ (see Definition 6). If the node is in the array, the process ends. Otherwise, it continues until $l$ reaches $L-1$.

%where $r$ is the exact number of output reactions of the in-process metabolite in depth $k$.}
%...........................................
%\input{files/alg1.tex}
\begin{algorithm}[!tpb]
\label{alg:one}
\caption{DEPTHy: Finding dependent reactions with depth $L$
over metabolite $\mathcal{N}_i$ on thread $T_j$ in a given $\mathbf{MG}$.}
\textbf{         }\\
\DontPrintSemicolon
\SetAlgoLined
\scriptsize
\KwIn{The $\mathbf{MG}$ graph of a metabolic network derived from its given stoichiometric matrix;
Metabolite $\mathcal{N}_i$; Thread ID $T_j$; Maximum number of input/output reactions of a metabolite $r_{max}$.}
\KwOut{Tagged dependent reactions of $\mathcal{N}_i$.}
\textbf{         }\\
Convert the decimal thread ID $j$ to an $L$-digit number in base $r_{max}$; i.e., $j$=$(L_{0\leq l < L})_{r_{max}}$=$(L-1, L-2,...,1,0)_{r_{max}}$\\
Store $\mathcal{N}_i$ in the queue $\mathcal{Q}$.\\
Set $l$ to 0.\\
Set $Start$ reaction as an input reaction of $\mathcal{N}_i$.\\
\While {$\mathcal{Q}\neq \emptyset$ and $l \leq L$}
{
    Select the output reaction at level $l$ based on $L_l$.\\
    Check if the output is valid according to the size of $O$ of $\mathcal{Q}_l$.\\
    \uIf{$\mathcal{Q}_l$ $\in$ $O_{M_{L_l}}$}
    {
        Set $O$ as $Last$ reaction.\\
        Tag the $Start$ reaction and the $Last$ reaction of $\mathcal{N}_i$ as dependent reactions.\\
        Exit \textbf{while}.\\
    }
     \Else
    {
        Push back  $O_{M_{L_l}}$ to $\mathcal{Q}$.\\
    }
}
\end{algorithm}
%...........................................

\textbf{Initialization of target reactions.}
In this function, the target reactions which are selected to contribute in the pathway are initialized with a default value (``1'' in our implementation). These reactions can be boundary reactions or any other set of reactions.
Connected metabolites to the initialized reactions are triggered by them.
%The initialized reactions trigger the metabolites connected to them.
To do this using the GPU, the pair ($\#$ of blocks, $\#$ of threads) is set as ($C$,$R$) where $C$ is the number of EFM candidate pathways and $R$ is the number of reactions of the network. This is done on the GPU in one pass. Each thread sets/resets the initial value of a reaction on each block which is dedicated to a candidate pathway.

\textbf{Calculation of balanced pathways.}
In META$x$ function in Algorithm~1, the pair ($\#$ of blocks, $\#$ of threads) is set as ($C$,$M$). Each block is dedicated to an EFM candidate pathway and each thread of the block is dedicated to a metabolite. The META$x$ function is set up on all threads. Algorithm~3 shows the function as sketched in the system design flow in Fig~\ref{fig:design}. In this function, each pathway is created using randomly selected output/input in the forward/backward flow.
EFM candidates at the end of the algorithm change to either \emph{Done} or \emph{notEFM} status.
The status remains as \emph{notDone} as long as there are still unbalanced metabolites on the pathway. Each thread executes a \emph{while(true)} loop. The while loop is broken when all candidates take a \emph{Done} or \emph{notEFM} status. The function of the arbiter module in Fig~\ref{fig:design} is embedded in the META modules performing META$x$ functions. This is because all threads should run the same function in each GPU access.

%...........................................
%\input{files/alg2.tex}
\begin{algorithm}[!tpb]
\label{alg:two}
\caption{METAx: on Thread $T_i$ of Block $B_j$.}
\textbf{         }\\
\DontPrintSemicolon
\SetAlgoLined
\scriptsize
\KwIn{The $\mathbf{MG}$ graph of a metabolic network derived from its given stoichiometric matrix;
Thread ID $T_i$; Block ID $B_j$.}
\KwOut{minimal flux-balanced pathways.}
\textbf{         }\\
\While {\emph{true}}
{
   \ForAll{Inputs and Outputs of $\mathcal{N}_{T_i}$}
   {
        Multiply the input/output coefficient to its flux value (i.e., $I_{c_{ij}}$$\times$$I_{f_{ijp}}$ and $O_{c_{ij}}$$\times$$O_{f_{ijp}}$ where $p$ is assigned to $B_j$) and add the value to \emph{ExtraFlux} as stated in Eq. (2).\\
   }
   \textbf{         }\\
   \textbf{         }\\
   %%........................
   \uIf{\emph{ExtraFlux} $=$ 0}
   {
    Tag $\mathcal{N}_{T_i}$ as a $Stable$ metabolite.\\
   }
   %%.............
   \ElseIf{\emph{ExtraFlux} $>$ 0}
   {
        \uIf{$\mathcal{N}_{T_i}$ is not $visited$}
        {
            Select a random output reaction.\\
            Calculate the flux using Eq. (3).\\
            Save \emph{Primary} input, $R_{PI}$, and \emph{Primary} output, $R_{PO}$, for $\mathcal{N}_{T_i}$.\\
            Tag the metabolite as \emph{visited}.\\
        }
        \Else
        {
            \uIf{$R_{PI}$$\propto$$R_{PO}$ is $\textbf{false}$ and \emph{MAX\_LOOP} is not reached}
            {
                Select output reaction $R_{PO}$ and use Eq. (4).
            }
            \ElseIf{$R_{PI}$$\propto$$R_{PI}$ is $\textbf{true}$ and \emph{MAX\_LOOP} is not reached}
            {
                Select input reaction $R_{PI}$ and use Eq. (4).
            }
            \ElseIf{\emph{MAX\_LOOP} is reached}
            {
                Tag the EFM candidate as \emph{notEFM}.

            }
        }
   }
   %%.............
   \ElseIf{\emph{ExtraFlux} $<$ 0}
   {
        \uIf{$\mathcal{N}_{T_i}$ is not \emph{visited}}
        {
            Select a random input reaction.\\
            Calculate the flux using Eq. (3).\\
            Save \emph{Primary} output and input for $\mathcal{N}_{T_i}$.\\
            Tag the metabolite as \emph{visited}.\\
        }
        \Else
        {
             \uIf{$R_{PI}$$\propto$$R_{PO}$ is $\textbf{false}$ and \emph{MAX\_LOOP} is not reached}
            {
                Select input reaction $R_{PI}$ and use Eq. (4).
            }
            \ElseIf{$R_{PI}$$\propto$$R_{PI}$ is $\textbf{true}$ and \emph{MAX\_LOOP} is not reached}
            {
                Select output reaction $R_{PO}$ and use Eq. (4).
            }
            \ElseIf{\emph{MAX\_LOOP} is reached}
            {
                Tag the EFM candidate as \emph{notEFM}.

            }
        }
   }
   \textbf{         }\\
   \textbf{         }\\
   %%........................
   \ForAll {Threads (Metabolites) in the Block $B_j$}
    {\If{ All metabolites are \emph{Stable} }
        {Tag the candidate as \emph{Done}.\\}
    }

    \ForAll {Blocks (EFM candidates)}
    {
     \If{All candidates tagged as \emph{Done} or \emph{notEFM}}
        {Exit \textbf{while}.\\ }
   }
}
\end{algorithm}
%...........................................

\subsection{Application on biologically relevant metabolic networks}\label{sec:appOnNetworks}
\subsubsection{CHO cell metabolism}
To prove the proposed concept on metabolic networks, a core model of Chinese Hamster Ovary (CHO) cell metabolism was chosen from~\cite{provost2004dynamic,jayapal2007recombinant}. CHO derived cell lines are the preferred host cells for the production of therapeutic proteins~\cite{xu2011genomic}. The stoichiometric matrix of this network is of size $12$$\times$$18$ with two boundary reactions $r_0$ and $r_{15}$. An NVIDIA GeForce GT330M GPU and CUDA version 5.5 platform were chosen. An Intel Core-i5 CPU with 4 GB RAM was used to run the test.
The source code and results are provided at https://github.com/marabzadeh/GPU-GB-EFM.
The network has 7 EFMs. Using random pathway creation, five EFMs were obtained. Table~\ref{tab:table2} summarizes the results after attempting 48 EFM candidates with depth $L=4$ and $r_{max}=4$. The global memory used to store the network on GPU was approximately 2 MBytes.
To test that all EFMs are computable by the approach, contributing reactions of each EFM were selected intentionally in the algorithm and their pathway were created. This test was passed and all EFMs were observed. Considering enough time and resource, the approach can lead to the calculation of all EFMs, as proved by Theorem 1.

\begin{table}%[!tpb]
\caption{The results for three kernel functions after running 48 EFM candidates with depth $L=4$ and $r_{max}=4$.
The columns \textbf{Local M} and \textbf{Shared M} report the memory usage in bytes in each thread and each block, respectively. The last column shows the number of registers allocated for each thread.}
\label{tab:table2} {
\centerline{
\scriptsize
%\small
%\tiny
\begin{tabular}{lcccccc}
\hline
Function Name	&$\#$of Candidates	&$\#$of Threads	&Duration ($\mu s$)	&Local M/Thread	&Shared M/Block &Registers/Thread \\
\hline
\hline
DEPTHy	&12	&256	&149.7		&816	&24	&23	\\	
INIT	&48	&18	&17.9		&0	&32	&12	\\
METAx	&48	&12	&15583.2	&60	&40	&27	\\
\hline
\end{tabular}}
}
\end{table}

Fig~\ref{fig:cho} illustrates the steps for creating one of the EFMs in the metabolic network of CHO cell. The table in the top-right corner of this figure shows the flux of reactions from the point of view of the in-process metabolites step by step. There are six steps to balance the pathways and Step 7 shows the fluxes of the balanced pathway.
The in-process metabolites in each step are shown by big-circles.
In each step, fluxes are not equal in active threads. This happens due to the accessibility of threads to the shared data of the reactions. To overcome the conflict, a flag is set when a reaction flux is read in Eq (1) and reset after write. If another thread reads the flag as set, it should wait for the flag to reset.

%...........................................
\begin{figure}[!h]
\centering
\includegraphics[scale=0.2]{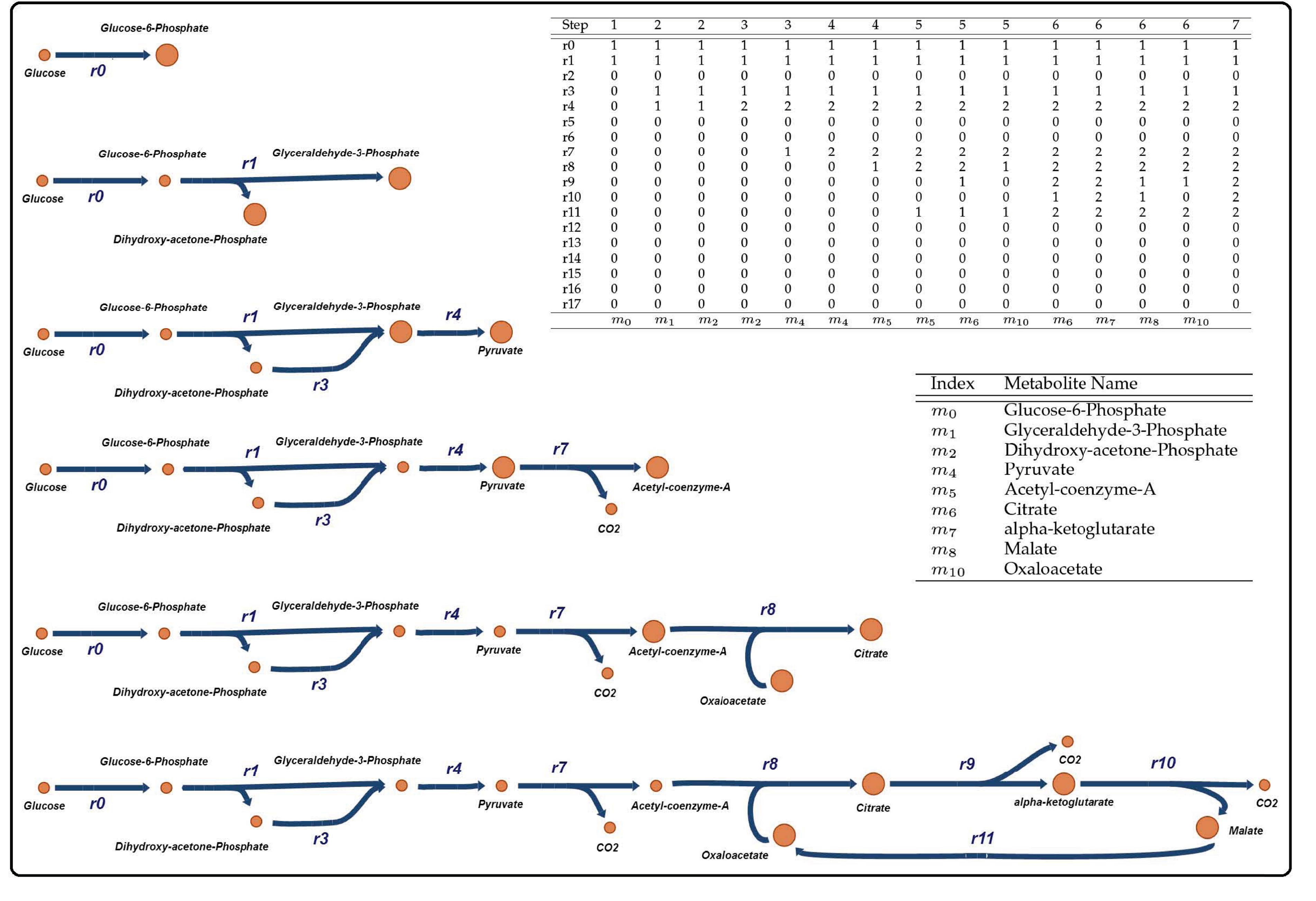}
\caption{{\bf An EFM creation in CHO metabolic network.}
The steps of creating one of the EFMs of the metabolic network of CHO cell. Six steps are required to balance the fluxes of the EFM. In each step, from top to down, in-process metabolites are shown with big-circles and stable metabolites are shown with smaller ones. The fluxes are shown in the table from the point of view of the in-process metabolites in each step. The last column, Step 7, is the final fluxes of the EFM identical from the view of all metabolites. Paths are depicted using Escher~\cite{king2015escher}}
\label{fig:cho}
\end{figure}
%...........................................
%\begin{figure*}[!tpb]
%\centering
%\includegraphics[scale=0.2]{files/cho.png}
%\caption{The steps of creating one of the EFMs of the metabolic network of CHO cell. Six steps are required to balance the fluxes of the EFM. In each step, from top to down, in-process metabolites are shown with big-circles and stable metabolites are shown with smaller ones. The fluxes are shown in the table from the point of view of the in-process metabolites in each step. The last column, Step 7, is the final fluxes of the EFM identical from the view of all metabolites.}
%\label{fig:cho}
%\end{figure*}
%...........................................

Table~\ref{tab:table3} reports some observations when different sets of reactions are initialized.
It should be mentioned that different starting reaction sets may potentially lead to different EFMs.
For each metabolite, the unbalanced status of the metabolite is counted. The number of conflicts is also reported. Subtracting the number of conflicts from the number of steps gives the actual number of parallel steps. The \emph{Sum} column shows the total number of times that all threads were activated and the \emph{Ave} column shows the average number of a thread is active, that is, the thread has the computational load of updating fluxes. The column \emph{Meta/Step} shows the average number of in-process metabolites at each step. As can be derived from the table, by initializing internal reaction(s), more steps are taken to make the pathway stable.

\begin{table}%[!tpb]
\caption{The operation of the system when different reaction sets are initialized. Three types of sets with one, two and three reactions are selected, in which only reaction 0 is a boundary reaction.}
\label{tab:table3} {
\centerline{
%\small
%\tiny
\scriptsize
\begin{tabular}{lccccccccccccccc}
\hline
Selected reactions	&Steps($\approx$)	&Conflict	&Step-Conf	&Sum	&Ave &Meta/step	&$m_0$	&$m_1$	&$m_2$	&$m_4$	&$m_5$	&$m_6$	&$m_7$	&$m_8$	 &$m_{10}$	\\
\hline
\hline
r0			&8		&1	&7	&16	&1.8	&2.3	&2	&1	&2	&2	&2	&2	&2	&1	&2		\\
r1			&8		&1	&7	&16	&1.8	&2.3	&2	&1	&2	&2	&2	&2	&2	&1	&2		\\
r7			&60		&14	&46	&67	&7.4	&1.5	&4	&3	&6	&4	&6	&14	&20	&4	&6		\\
r9			&61		&14	&47	&77	&8.6	&1.6	&4	&3	&6	&4	&6	&14	&19	&15	&6		\\
r1,10			&57		&14	&43	&77	&8.6	&1.8	&4	&3	&6	&4	&6	&14	&19	&14	&7		\\
r4,11			&59		&14	&45	&77	&8.6	&1.7	&4	&3	&6	&4	&6	&15	&19	&14	&6		\\
r1,7,9			&55		&16	&39	&78	&8.7	&2.0	&4	&3	&7	&4	&7	&13	&19	&15	&6		\\
r0,3,8			&56		&15	&41	&78	&8.7	&1.9	&5	&3	&6	&4	&6	&14	&20	&14	&6		\\
\hline
\end{tabular}}
}
\end{table}

\subsubsection{\emph{E.}\emph{coli} model iAF1260}
To analyze the applicability of the proposed architecture on other metabolic networks, the model of \emph{E.}\emph{coli} network with 1668 metabolites and 2382 reactions was chosen from~\cite{king2015bigg}. ATP, AMP, ADP, H, Pi, CO$_2$, H$_2$O, COA, NAD, NADP, NADH and NADPH are considered as cofactors. As stated in the literature~\cite{de2009computing}, some simplifications to the system\textquotesingle s model can be performed in order to reduce the complexity of the problem, such as setting currency metabolites like cofactors and energy metabolites to external. According to~\cite{de2009computing}, for energy currency metabolites like ATP, NADH and FADH, since their concentration is assumed to be constant, they are not required to be balanced by an EFM. Notations are taken of the COBRA model from~\cite{king2015bigg}.

The set of shortest EFMs producing L-Lysine calculated by~\cite{de2009computing} was extracted. The observability of 7 shortest EFMs with length 27 and 28 was studied on the proposed architecture. Fig~\ref{fig:ecoliP1} illustrates
the steps of constructing the EFM depicted in the figure. At each step, the set of unbalanced metabolites are specified and the required changes on the reaction set are illustrated. Reaction fluxes are changed so as to achieve the least required steps to create and balance the targeting EFM.

%...........................................
\begin{figure}[!h]
\centering
\includegraphics[scale=0.32]{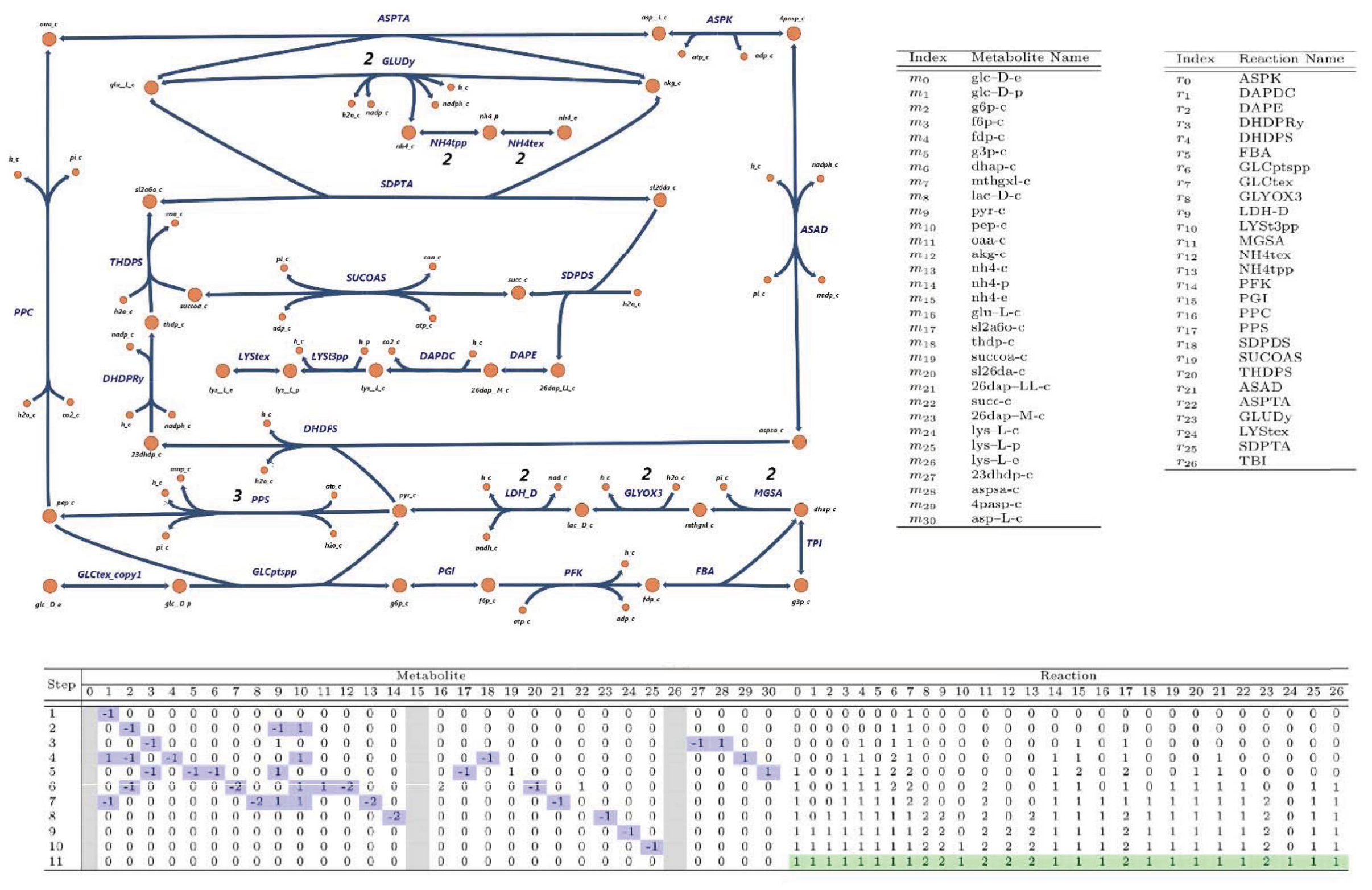}
\caption{{\bf Producing an EFM in \emph{E.}\emph{coli} metabolic network.}
Eleven steps are required to balance EFM fluxes. At each step, the set of unbalanced metabolites are specified by blue cells. The process of these cells can run simultaneously. Flux changes for reactions are shown at each step. The last row states that there is no unbalanced metabolite and the flux set indicates the EFM fluxes in the steady state. The EFM is depicted using Escher~\cite{king2015escher}. Metabolite and reaction notations are taken from the model of~\cite{king2015bigg}. }
\label{fig:ecoliP1}
\end{figure}
%...........................................
%...........................................
%\begin{figure*}[!tpb]
%\centerline{
%\includegraphics[scale=0.18]{files/ecoliP1.png}}
%\caption{.}
%\label{fig:ecoliP1}
%\end{figure*}
%%...........................................

The same information for all 7 paths are provided in Supplementary material File 1 (S1.exe) and the resulting EFMs are depicted in Supplementary material File 2 (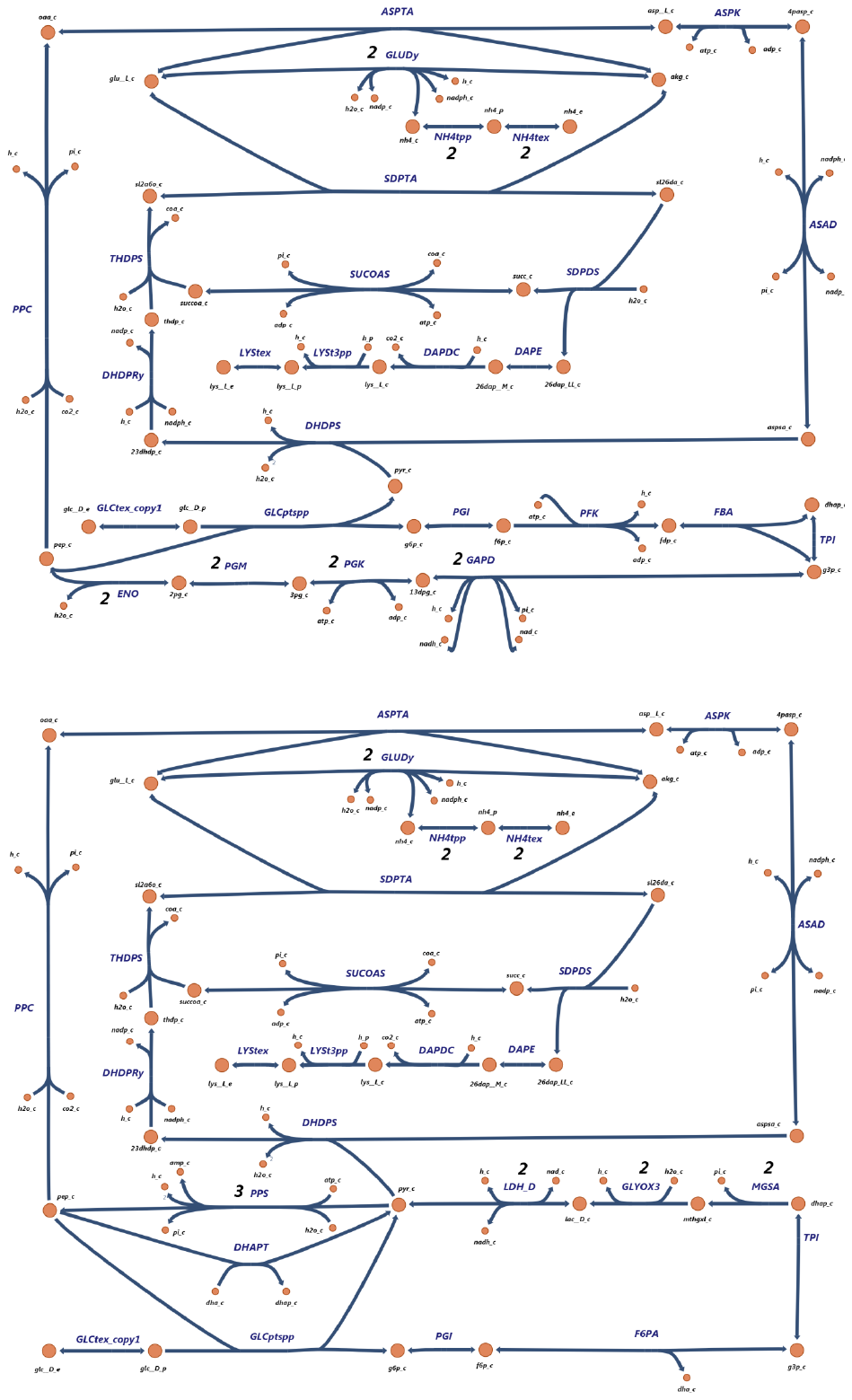). Summary of the results are provided in Fig~\ref{fig:ecoliInf}. As shown in Fig~\ref{fig:ecoliInf}-a, the number of unstable metabolites tends to match a normal curve which means that the number of unbalanced metabolites is increased in the intermediate steps. The box plots are also depicted in Fig~\ref{fig:ecoliInf}-b for each path to show the distribution and the average of the unbalanced metabolites.

%...........................................
\begin{figure}[!h]
\centering
\includegraphics[scale=0.5]{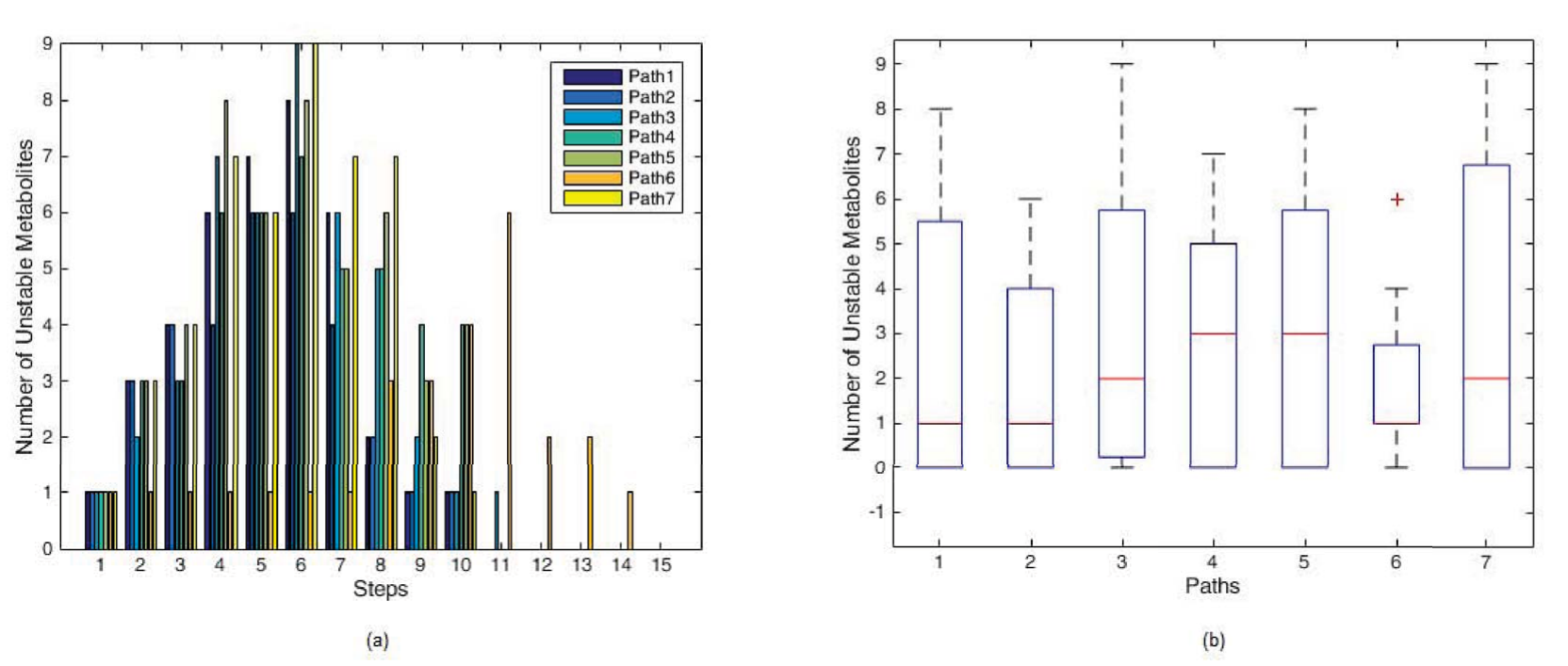}
\caption{{\bf Parallel simulation of shortest EFMs of \emph{E.}\emph{coli} iAF1260.}
(a) The number of unbalanced metabolites at each step for 7 pathways. (b) The distribution of unbalanced metabolites at each step for each pathway.}
\label{fig:ecoliInf}
\end{figure}
%%...........................................
%\begin{figure*}[!tpb]
%\centerline{
%\includegraphics[scale=0.5]{files/ecoliInf.png}}
%\caption{(a). (b).}
%\label{fig:ecoliInf}
%\end{figure*}
%%...........................................

This study shows the potentiality of the proposed architecture for finding EFMs on large-scale metabolic networks. Three sets of variables influence the generation of an EFM, specially on large-scale metabolic networks: (1) a set of first-initialized reactions, (2) the primary input and output selection at first-visited metabolites, (3) the decision on how to pass the flow on the metabolites which are not first-visited. In this test, the external input reaction is initialized and for the rest, the best decisions were taken to pass the flow through one of the reactions of the in-process metabolite.
In this example, an exhaustive search is applied to find the best selection to pass the flow. However, to make the approach automated for large-scale networks, further studies are required.
The order of metabolites and their accessibility to the reaction\textquotesingle s shared data and the decision on how to pass a flow to balance a pathway or report it as an unstable path is important. The convergence of the solution is discussed in Section \emph{Theoretical analysis}. However, a practical implementation to reach the desired solution in a reasonable time can be proposed by further research. An example is provided in Supplementary material File 4 (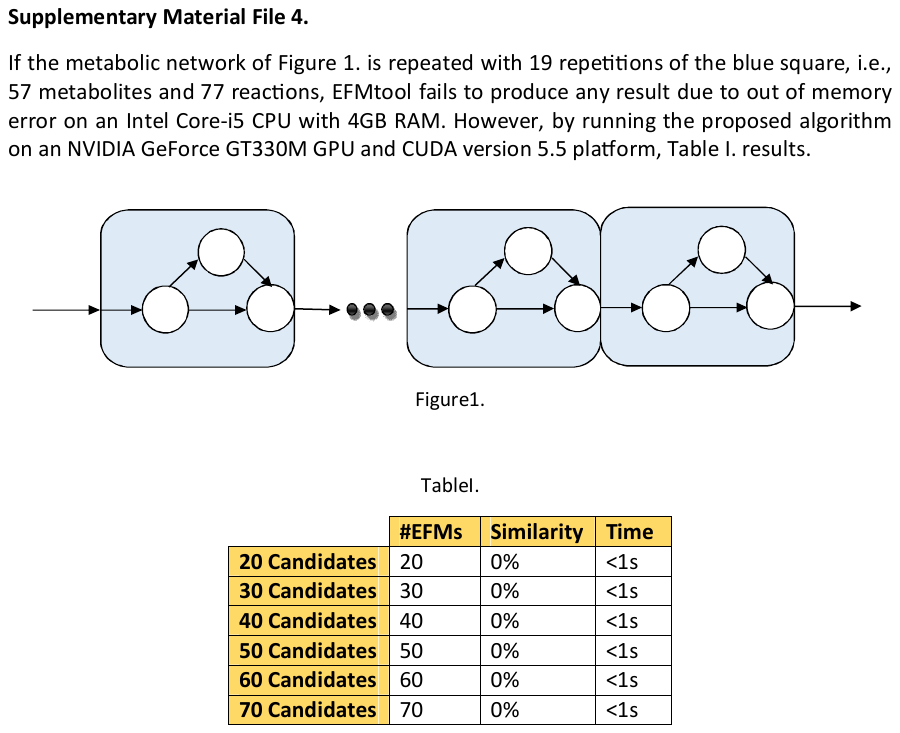) to illustrate the benefit of linear random pathway construction.

%%-----------------------------------------------------------------------------------------------------------------------------------
%%-----------------------------------------------------------------------------------------------------------------------------------

\section{Technical details}\label{sec:TeDe}
%...........................................
\subsection{Design parameters}
In parallel computing, several design parameters should be considered to make the system work efficiently on the hardware platform. Some of these parameters and the way they are treated in the proposed design are discussed.\\

\textbf{Granularity.} The level of details considered in a model or decision making process is considered as \emph{Granularity}. The coarser granularity results in the deeper level of details. Granularity in parallel systems is defined as the size of a part of a system which is selected to work independently.
In the proposed design, the granularity is chosen as the size of a metabolite according to the data stored in each metabolite and the instructions to be executed in each of them. Besides, the function of a metabolite is the smallest repeatable part of the system.

%\textbf{Load Balancing and Locality.} Each part of the system has a timing cost of initializing the process and loading the required data.
%\red{Large unit of works reduces the overhead of parallelism but also reduces the available parallelism as well.}

\textbf{Coordination and Synchronization.}
%\red{The cost of communication between shared data and the cost of synchronizing are as a part of coordination and synchronization subject matter.}
The memory and timing cost of data communication and resolving inconsistencies, which may occur because of data sharing, should be considered in the design.
Uniform shared memory over different blocks and message passing are among data coordination methods.
In the proposed design, the hierarchical memory structure with different sharing levels is used for data communication. To overcome the conflicts, solutions are proposed as discussed earlier.

%levels are considered and shared memory is used as discussed earlier.

%%------------------------------------------------
%...........................................
\subsection{Topology-based parallelism analysis}
In the proposed model, two levels of parallelism are used as explained below.
\subsubsection{Block-level (pathway-level) parallelism}
This level of parallelism is based on the independence of different pathways in the network.
In order words, distinct pathways are created from different outputs of a metabolite. Each new pathway is independent of the others.
In our implementation, each block is dedicated to a pathway which executes independently. The idea is shown in Fig~\ref{fig:parallStru}-a for a metabolite with two reactions.

\subsubsection{Thread-level (metabolite-level) parallelism}
This level of parallelism is based on the node structure. When a reaction takes an updated flux, all connected nodes to that reaction, except the one which was in-process, are triggered. In the proposed implementation, the metabolite-level parallelism is applied by performing a META$x$ function on different threads of a block. The idea is shown in Fig~\ref{fig:parallStru}-b for a reaction with four connected metabolites.

%...........................................
\begin{figure}[!h]
\centering
\includegraphics[scale=0.7]{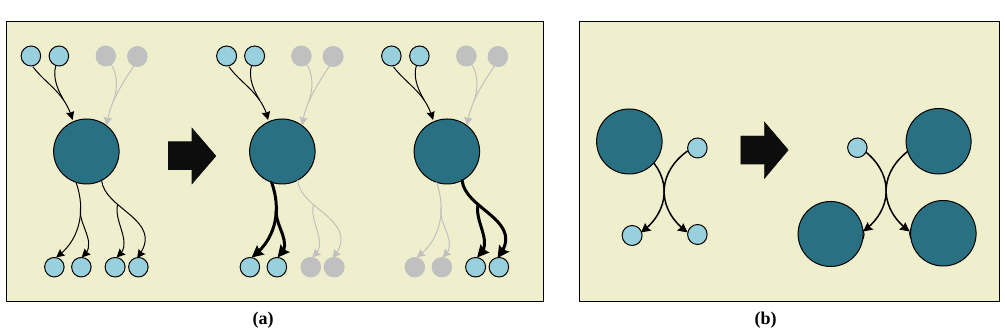}
\caption{{\bf Parallelism levels.}
(a) Block-level parallelism. The two pathways can be traversed independently. (b) Thread-level parallelism. The three metabolites can be processed simultaneously. Big circles are nodes which are getting active concurrently. Thick arrows represent simultaneous active paths.}
\label{fig:parallStru}
\end{figure}
%...........................................
%\begin{figure*}[!tpb]
%\centering
%\includegraphics[scale=0.7]{files/parallStru}
%\caption{(a) Block-level parallelism. The two pathways can be traversed independently. (b) Thread-level parallelism. The three metabolites can be processed simultaneously. Big circles are nodes which are getting active concurrently. Thick arrows represent simultaneous active paths.}
%\label{fig:parallStru}
%\end{figure*}
%...........................................

\subsubsection{Performance modeling}
It is usually hard to model the performance of a parallel system. Modeling the performance of the proposed design is even harder since each thread executes in a while loop either to make a pathway flux-balanced or discard it.
$Z_n$ function in Eq~\ref{eq:6} is defined to model the number of active parallel metabolites (threads) at Step $n$ of the system.

\begin{equation}\label{eq:6}
\begin{array}{l}
Z_1 = 1,\\
Z_n = Z_{n-1} \times P_{rM}(rM_n) \times P_{mR}(mR_n - 1).
\end{array}
\end{equation}

In this equation, $rM_n$ is the number of input or output reactions at each step and $mR_n$ is the number of metabolites involved in a reaction. At step $n$, for each metabolite there are $rM_n$ reactions and each reaction has $mR_n - 1$ triggered metabolites. Going through the steps, the dependencies in the graph structure cause the nodes to be visited more than once and pathways to be duplicated. $P_{rM}$ and $P_{mR}$ functions are used in the model to consider the graph dependencies in terms of pathway and metabolite, respectively. $P_{rM}$ and $P_{mR}$ are probability functions which are estimated to have high probabilities at first and then converge to zero. These functions can be modeled by more experiments on metabolic networks.
A CPU-GPU comparison analysis based on this performance model is provided in Supplementary material File 3 (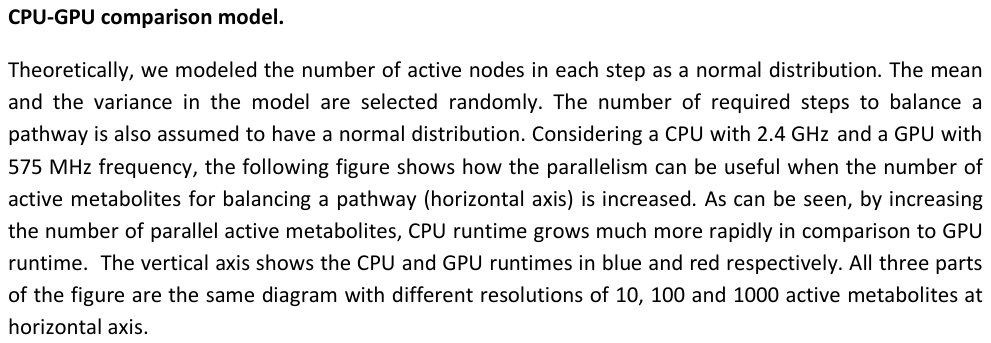).

% with descending orders through steps.
%...........................................
\subsection{Theoretical analysis}
In this section, the proposed system architecture for finding EFMs is discussed theoretically on the graph model.

\begin{prop}
By initializing a selected set of reactions $R_{in}$=$\{r_k|k \in indices\}$ with a default flux value, all minimal pathways which are including $R_{in}$ are explored in the graph $\mathbf{MG}$ of a given metabolic network; the set ``indices'' indicates a subset of reactions in $\mathbf{MG}$.
\end{prop}

A set of output and and-related metabolites of the reactions in the set $R_{in}$ are triggered at each step $s$ by those reactions. From each first-visited metabolite, new pathways are created through outputs/inputs in forward/backward flow getting \emph{primary} tags. Recursively, all pathways including the set $R_{in}$ are created until there are no metabolites to be triggered for each path.

Since an arbitrarily selected set of reactions can be given, disconnected subgraphs might be generated and the given selection may violate the elementarity of the generated path. Therefore, undesired paths should be removed from the results.

%======================================
\begin{prop}\label{prop:step3}
Starting with an initialized set of reactions with certain fluxes, all nodes are traversed and all reactions in a pathway are assigned a flux using Eq.~\ref{eq:2}, as each node has a primary input and a primary output to be used in this equation. Equation~\ref{eq:2} states that the amount of the incoming and outgoing fluxes in $\mathcal{N}_i$ should be equal.
\end{prop}
%======================================
\begin{prop}\label{prop:step4}
Flux dependencies can prevent the rate of the production/consumption of an internal metabolite in a given topology of a pathway from being zero. Eq.~\ref{eq:3} is used to calculate an update for the consumption/production rate of the node. In the following cases, the pathway is discarded:
\begin{itemize}
  \item $I_{f(n)_{ijp} }=0$ or
  \item $O_{f(n)_{ijp} } = 0$,
  \item a loop is repeated over a node (i.e., getting back to a node from the same reaction multiple times), while trying to find a way out in a subgraph of the pathway.
\end{itemize}
\end{prop}

%======================================
\begin{lemma}\label{lemma:gra}
If a pathway is EFM, its graph model is reachable by graph traversal of Proposition 1.
\end{lemma}

\begin{proof}
An EFM is a flux-balanced pathway with no flux-balanced subset. Node-dependent pathways are referred to the pathways in which all related nodes of a selected reaction are traversed. All node-dependent pathways are produced according to Proposition 1 through different reaction choices. Therefore, either a flux-balanced pathway has a flux-balanced subset (which is node-dependent and reachable by Proposition 1) or it is an EFM itself and reachable by Proposition 1 as well.
\end{proof}
%======================================
%If S is well-conditioned or a condition on the number of related reactions or so what there exist a unique solution underdetermind homogenous what constraint over S to have at least a solution.

\begin{lemma}\label{lemma:cons}
In a consistent network, for two types of pathways there exists a flux-balanced solution. \emph{Type I} pathways are defined as pathways with no hyperarcs; either cycles or pathways from an input to an output with one input and one output for each node. \emph{Type II} pathways are those with nodes with more than one input or more than one output but with no flux-dependent reactions.
\end{lemma}

\begin{proof}
For \emph{type I} pathways, all nodes have one primary input and one primary output. Using Eq.~\ref{eq:2} subsequently, as stated in Proposition 2, fluxes are assigned to all reactions such that consumption/production of metabolites are balanced. For \emph{type II} pathways, using Eq.~\ref{eq:3}, as stated in Proposition 3, for all visited nodes there is a sequence of reactions through primary inputs/outputs to balance production/consumption of metabolites.
\end{proof}

To analyze pathways with flux-dependent nodes, an optimization problem is defined in Lemma~\ref{lemma:opt}.
%======================================

%======================================

\begin{lemma}\label{lemma:opt}
Finding flux values for potentially flux-balanced pathways in the proposed architecture can be defined as an optimization problem.
\end{lemma}

\begin{proof}
Function $f$ for node $i$ is defined as $f(i)=\sum {I_{c_{ij} }I_{f_{ijp} }  }  - \sum {O_{c_{ij} }O_{f_{ijp} }  }$. The problem of finding flux-balanced pathways, considering that the function of each metabolite $i$ is independent, is defined as follows:\\\\
Minimize:  $g(i)=\sum{f(i)}$,  $i$$=$$1$$:$$M$, in which $M$ is the number of active metabolites in the candidate pathway $p$.\\\\
Constraints: For each node $i$, the outgoing/ingoing output/input reaction should be selected such that $g(i)=0$. The selected reaction carries the flux from the set of active reactions $j$ in $p$ as stated in Propositions 2 and 3.
\end{proof}

%======================================
%======================================
\begin{theorem}
Sampling parallel pathways on the proposed system architecture, leads to exploration of a set of EFMs which includes all \emph{type I} and \emph{type II} EFMs.
\end{theorem}

\begin{proof}
Based on the result of \emph{coupon collector\textquotesingle s} problem in probability theory, the expected number of picks required to choose all the elements of a set is $n\sum\limits_{k = 1}^n {\frac{1}{k}}$ which for large $n$ is approximately $n\log n$. Consider the length of a pathway as the number of its contributed metabolites. For a pathway with length $l$, for each metabolite, $r_{max}$ selections are possible, in which $r_{max}$ is the maximum size of output/input array in the forward/backward flow. Therefore, ${(r_{max}\log r_{max})}^l$ efforts for a pathway with length $l$ are required to create all possible routes. Number of efforts reflects number of parallel pathways, considered as distinct blocks on the GPU platform. The result of Lemma~\ref{lemma:gra} is used to show that by selecting all outputs/inputs (in forward/backward flow), all pathways are produced. The results of Lemma~\ref{lemma:cons} and Lemma~\ref{lemma:opt} are used to show that the function running on each GPU thread keeps metabolite production/consumption rates balanced and leads to the construction of \emph{type I} and \emph{type II} flux-balanced pathways. With a proper solution for the optimization problem of Lemma~\ref{lemma:opt}, a balanced solution for other pathways can be explored.
\end{proof}
%======================================
\section{Discussion}\label{sec:dis}
The proposed model traverses the graph and keeps both the stoichiometry information and the minimality of the path with the opportunity of not exploring the whole solution space. Our approach tries to decide if a certain ``path'' is EFM or not, merely based on topology rules (instead of recognizing the elementarity through rank test or comparing new candidates with produced ones). The graph structure makes a good track of unbalanced metabolites.
Each metabolite is considered as an object and the graph structure makes it easier to set-up rules for paths and explore the intended solution space via rules; e.g., initializing a set of reactions to find their including EFMs.

%A hypergraph is a simplified type of the modified AND/OR graph. In the introduced model, each arc has two different input and output coefficients. Besides, for each arc, a dynamic label for each pathway is defined, that is, the flux of the reaction represented by that arc.
%The complication arises since fluxes of contributing input and output reactions over a node may be related. This results to the non-linear property of fluxes. The non-linearity has been considered in our flux calculation procedure (see Eq. 2 and Eq. 3).

The modified AND/OR graph is a simplified type of a hypergraph and the application of hypergraphs in biological networks has been reviewed in~\cite{klamt2009hypergraphs}. In the proposed model, each arc has two different input and output coefficients. Besides, for each arc, a dynamic label for each pathway is defined, that is, the flux of the reaction represented by that arc~\cite{arabzadeh2017graph}. In this context, the edges are directed and each has been associated with a weight. However, the weight of an output arc of a node may be different from the weight of the arc entering another node as an input. The reason is that the production rate of a metabolite in one reaction can be different from its consumption rate in another reaction. It has been studied that the hypergraph structure complicates pathway topology~\cite{marashi2014mathematical}. The complication arises since fluxes of contributing input and output reactions over a node may be related. This results to the non-linear property of fluxes. The non-linearity has been considered in our flux calculation procedure (see Eq. 2 and Eq. 3). Besides, since in this model, each metabolite is considered as a subsystem, the linear or non-linear relations of reactions over a metabolite, i.e. a subsystem, are of interest as defined in Definition 11. Other definitions such as enzyme subsets as defined in~\cite{pfeiffer1999metatool} are not directly related to our definitions and Definition 11 can be considered as a subset of enzyme subsets. Enzyme subsets are defined in~\cite{pfeiffer1999metatool} as groups of enzymes that operate together in fixed flux proportions in all steady states of the system over a whole pathway.

In comparison to the method of~\cite{de2009computing}, both methods are exploring the non-steady-state solution space and try to apply additional constraints to reduce the complexity. However, in this approach instead of using the support of LP-based tools, we used an inherent parallel structure of the network and developed a model on hardware. However, while the basic framework is sketched and the potential of the method on pathways belongs to large networks is shown, further studied are required to automatize the approach to get desired output EFMs.

The first introduced double-description method for finding EFMs explores the whole set of pathways to find pathways that are in the steady-state solution space and then selects EFMs by comparing the reaction subset of pathways~\cite{schuster1994elementary} as used in several tools such as COPASI~\cite{hoops2006copasi}. In the improved double-description versions of the method which use null-space of the stoichiometry matrix as an input, different combinations of pathways in the null-space are calculated and then the reaction subset of pathways are compared~\cite{wagner2004nullspace,urbanczik2005improved,quek2014depth}. In~\cite{arabzadeh2017graph}, GB-EFM method first calculates the reaction-dependent pathways in the solution space and then checks to see if these pathways can be in the steady-state by using some rules on the topology of the pathway. In the proposed model the procedure of constructing the paths and balancing them are combined together to introduce a hardware independent core so as to combine the results of all these cores with the same decision table and distributed data on the network to calculate EFMs.

There are two different categories of approaches for constraint-based analysis, namely FBA and EFM analysis. Both approaches use the concept of ``balanced fluxes''. However, in FBA, LP is applied to find fluxes in metabolic networks with an optimized objective function. In contrast, in our approach, which is in the category of EFM analysis approaches, we use the concept of ``balanced fluxes'' in a graph representation of a metabolic network to find flux-balanced pathways. To be more precise, we did not use the LP-based methods like flux balance analysis (FBA) as used in~\cite{riemer2013metano} and~\cite{chan2018accelerating}. In contrast, we use an AND/OR graph to find pathways with specific properties.

Using acceleration methods to speed-up the calculation of computational biology algorithms has been a challenge for years~\cite{chan2018accelerating}. However, in our proposed method we aimed to introduce a system modeling rather than just an acceleration technique.

Several models have been introduced so far for the analysis of metabolic networks including convex analysis~\cite{pfeiffer1999metatool} and topological analysis~\cite{zevedei2003topological}. While all these approaches persue the same goal, there are important differences in the way they can be interpreted and be useful for a specific application.
Considering a system as a collection of subsystems has been modeled using cellular automata and were used in several applications~\cite{silva2003cellular,wishart2005dynamic}. The relation of graph models and cellular automata has been studied before~\cite{martinez2018simple}. However, to the best of our knowledge, in constraint-based approaches, the system view of a metabolic network, that is considering a metabolic network as a set of subsystems to analyze the network, has not been used before. To avoid the complexity of such modeling as cellular automata and because the main goal of the introduced method was to propose a hardware parallel model, our method is proposed in terms of graph objects and the relationship between subsystems defined based on flux calculation. However, the correspondence of the proposed method with cellular automata and using GPU in the context of topology-based parallelism to implement cellular automata (which leads to an acceleration tool to implement it) can be considered as a future work.

For computing minimal pathways of metabolic networks, some approaches, like MinSpan~\cite{bordbar2014minimal}, try to find the sparsest linear basis of the null space of the stoichiometric matrix, \textbf{S}. However, finding only the shortest pathways (as in~\cite{de2009computing}) may introduce a bias in computing EFMs, since in large-scale networks the majority of EFMs are known to be the large ones~\cite{machado2012random}. Furthermore, approaches which consider computing the shortest EFMs typically formulate their problem as a mixed-integer linear program (MILP), which is known to be NP-hard. More specifically, these algorithms find EFMs one by one, which means that after finding an EFM by solving an MILP, a new ``not-equal'' constraint should be added to the set of constraints in order to avoid finding the previous EFMs. Such constraints often make the latter MILP problem harder to solve~\cite{de2009computing}. In contrast, our approach is unbiased, in the sense that all EFMs, including the long ones, have the same chance to be found. Moreover, our approach applies parallelism for computing EFMs. Therefore, the EFMs that are found are independent of each other.

The proposed approach relies on availability of an array of hardware resources to achieve its goal. Taking advantage of the inherent structural parallelism of the hardware resources, we attempt to find as many EFMs as possible given the available hardware resources, i.e., the more parallel resources are available, the more EFMs may be found. Therefore, the user is not ``directly'' setting a limit on the number of EFMs. However, their choice of hardware resources will ``indirectly'' determine that.
Since the hardware platform features a parallel structure, the time complexity of the proposed approach is very low. Needless to say, in the latter case, expanding hardware resources would resolve the impasse and provide a set of EFMs. We believe that overcoming this obstacle and the related analysis can be the topic of another research.

%%-----------------------------------------------------------------------------------------------------------------------------------
%%-----------------------------------------------------------------------------------------------------------------------------------
\section{Conclusion}\label{sec:conc}

In this paper, a modular system architecture was proposed to calculate minimal flux-balanced metabolic pathways. The architecture is based on the AND/OR graph model. Each META$x$ module was designed in order to emulate the internal function of a metabolite for finding EFMs.
The proposed architecture was implemented on a GPU platform to take advantage of the parallel architecture provided in the GPU based on multiple cores and hierarchical memory. The memory levels of the GPU are used to illustrate the memory hierarchy in the system. The topology-based parallelism obtained by the system was the main achievement of the model.
Additionally, the simplified metabolic network of the CHO cell was studied to prove the concept of the design on metabolic networks to find EFMs. Besides, the potential of the model was studied on shortest pathways of the \emph{E.}\emph{coli} model.

Studying genome-scale models and finding biologically meaningful pathways~\cite{hadadi2017reconstruction,huang2017method} by setting rules in the module\textquotesingle s function are considered as our future research. In this paper, the static structure of the GPU was used. Using dynamic thread activation, as provided in recent GPU architectures, and store accurate pathway information according to the available memory space are the ideas to make the model appropriate for genome-scale analysis.
In addition, core function decisions can be improved by randomized decisions while keeping a global cost function to manage the moves, which is considered to investigate in our future research.
Besides, partitioning and compression preprocessing methods can be used further to overcome the limitation of the hardware platforms.

%%-----------------------------------------------------------------------------------------------------------------------------------
%%-----------------------------------------------------------------------------------------------------------------------------------
%%==================================================================================================================================================
%%==================================================================================================================================================
%%==================================================================================================================================================
%%==================================================================================================================================================
%%==================================================================================================================================================
\section*{Supporting information}

% Include only the SI item label in the paragraph heading. Use the \nameref{label} command to cite SI items in the text.
\paragraph*{S1 excel file.}
\label{S1_File}
{\bf Applying the proposed model on \emph{E.coli} shortest paths.} Complementary data are provided in the excel file.

\paragraph*{S2 pdf file.}
\label{S2_File.}
{\bf Applying the proposed model on \emph{E.coli} shortest paths.} Complementary figures are provided in the pdf file.

\paragraph*{S3 pdf file.}
\label{S3_File.}
{\bf CPU-GPU comparison model.} Complementary figure is provided in the pdf file.

\paragraph*{S4 pdf file.}
\label{S4_File.}
{\bf Benefit of linear random pathway construction.} Complementary data is provided in the pdf file.

\section*{Acknowledgments}
Authors would like to thank Dr. Nathan Lewis (UCSD) for helpful discussions on the concept.

\nolinenumbers

% Either type in your references using
% \begin{thebibliography}{}
% Text
% \end{thebibliography}
%
% or
%
% Compile your BiBTeX database using our plos2015.bst
% style file and paste the contents of your .bbl file
% here. See http://journals.plos.org/plosone/s/latex for
%\bibliography{biblio}

\end{document}